\documentclass[journal, twocolumn]{IEEEtran}

\newif\ifOneCol
\OneColfalse

\newif\ifIEEE
\IEEEtrue

\ifIEEE
\usepackage{amsmath,amsthm}
\fi
\usepackage{mathrsfs,amssymb,amsfonts,url,bbm,enumerate}
\usepackage{graphicx, subfigure, color}
\usepackage[draft]{fixme}
\usepackage{enumitem}
\usepackage{hyperref}
\usepackage[nodayofweek]{datetime}

\ifIEEE
\theoremstyle{plain}
\newtheorem{theorem}{Theorem}
\newtheorem{lemma}{Lemma}
\newtheorem{proposition}{Proposition}

\newtheorem{conjecture}{Conjecture}
\newtheorem{claim}{Claim}

\theoremstyle{definition}
\newtheorem{definition}{Definition}
\newtheorem{remark}{Remark}
\newtheorem{experiment}{Experiment}
\newtheorem{example}{Example}
\newtheorem{openprob}{Open Problem}
\newtheorem{question}{Question}

\else 
\theoremstyle{plain}

\theoremstyle{definition}
\newtheorem{remark}{Remark}

\newtheorem{openprob}{Open Problem}

\newtheorem{claim}{Claim}
\fi

\newif\ifCommentsOn
\CommentsOnfalse   

\ifCommentsOn  
\fxsetup{ status=draft,   author=,    layout=inline,    theme=color}
\definecolor{fxwarning}{rgb}{0.8000,0.0000,0.0000} 
\renewcommand{\comment}[1]{ \textcolor{red}{\bf{[[}}\fxwarning{#1}\textcolor{red}{\bf{]]}} }
\else 
\renewcommand{\comment}[1]{}
\fi

\newcommand{\hide}[1]{}

\renewcommand{\bar}[1]{\overline{#1}}
\renewcommand{\tilde}[1]{\widetilde{#1}}

\newcommand{\argmin}[1]{\underset{#1}{\operatorname{argmin}}}

\newcommand{\R}{\mathbb{R}}

\newcommand{\Ltwo}[1]{\big\lVert #1 \big\rVert}
\newcommand{\actionset}{\mathcal{M}}
\newcommand{\maxactset}{\mathcal{M}}
\newcommand{\divconst}{{C}}  
\newcommand{\dir}{\xi	}
\newcommand{\conv}{\mathrm{Conv}}
\newcommand{\maxpert}{\theta}
\newcommand{\Kmax}{K_{\max{}}}
\newcommand{\region}{R}
\newcommand{\regionset}{\mathcal{R}}
\newcommand{\x}{x}
\newcommand{\ptraj}{{\tilde{x}}}

\newcommand{\trend}{w}
\newcommand{\oneoverdel}{\gamma}
\newcommand{\CP}{\mathcal{C}}
\newcommand{\saferad}{\eta}
\newcommand{\basin}{\rho}
\newcommand{\ball}{\mathcal{B}}

\newcommand{\basinmin}{\basin_{\min{}}}
\newcommand{\ctresh}{\maxpert^*}
\newcommand{\imax}{{i^*}}
\newcommand{\Span}{\mathrm{span}}

\newcommand{\pert}{u}    
\newcommand{\Prt}{U}     
\newcommand{\ff}{\zeta}  

\newcommand{\measurable}{{\cal{L}}_{n}^1}
\newcommand{\ldset}{\mathcal{U}}	 
\newcommand{\Rone}{r}
\newcommand{\Rtwo}{\rho}
\def\xt{{\tilde x}}

\long\def\del#1{{\color{red}#1}}
\usepackage[normalem]{ulem}
\long\def\del#1{{\color{blue}\sout{[#1]}}}
\long\def\green#1{{\color{green}#1}}

\long\def\pur#1{{\color[rgb]{.8,0,.8}#1}}
\long\def\oli#1{{\color[rgb]{0,.8,.8}[#1]}}
\newcommand\redsout{\bgroup\markoverwith{\textcolor{red}{\rule[0.5ex]{2pt}{.5pt}}}\ULon}
\long\def\dela#1{{\color{blue}\redsout{[#1]}}}

\long\def\comm#1{{\color{green}[ #1 ]}}

\title{Sensitivity to Cumulative Perturbations\\ for a Class of Piecewise Constant Hybrid Systems}

\author{ 
Arsalan Sharifnassab$^\dagger{} ^*$,
\thanks{$\dagger$ A. Sharifnassab and S. J. Golestani are with the Department of Electrical Engineering, Sharif University of Technology, Tehran, Iran; email: asharif@mit.edu, golestani@sharif.edu}
John N. Tsitsiklis$^\ddagger$, \emph{Fellow, IEEE},
\thanks{$\ddagger$ J. N. Tsitsiklis is with the Laboratory for Information and Decision Systems, Electrical Engineering and Computer Science Department, MIT, Cambridge MA, 02139, USA; email: jnt@mit.edu}\\
 S. Jamaloddin Golestani$^\dagger$, \emph{Fellow, IEEE}
\thanks{$*$ This work was done while A.\ Sharifnassab was a visiting student at the Laboratory for Information and Decision Systems, MIT, Cambridge MA, 02139, USA. This work  was partially supported by the Iran National Science Foundation.}
}

\begin{document}
\maketitle

\begin{abstract}
We consider a class of continuous-time hybrid dynamical systems 
that correspond to subgradient flows of a piecewise linear and convex potential function with finitely many pieces, and which include the fluid-level dynamics of the Max-Weight scheduling policy as a special case.
We study the effect of an external  disturbance/perturbation on the state trajectory, and establish that the magnitude of this effect can be bounded by a constant multiple of the \emph{integral} of the perturbation. 
\end{abstract}

\section{\bf Introduction}\label{s:intro}
We consider a class of continuous-time, non-expansive, hybrid systems that are subject to an external disturbance/perturbation, and develop a bound on the effect of the perturbation on the state trajectory, in terms of the {\bf integral} of the perturbation.

In order to appreciate the issues that arise, and the usefulness of such a result, let us consider a discrete-time system of the form $x(t+1)=f\big(x(t)\big)$, $t=0,1,\ldots,$ 
and its perturbed counterpart 
\begin{equation} \label{eq:dt3}
\xt(t+1)=f\big(\xt(t)\big)+u(t),\qquad  t=0,1,\ldots
\end{equation}
Here, $x(t)$ and $u(t)$ take values in $\R^n$ and we assume that the mapping $f:\R^n\to\R^n$ is non-expansive, in the sense that 
$$\|f(x)-f(y)\| \leq \|x-y\|,\qquad \forall \ x,y\in\R^n,$$
for a given norm $\|\cdot\|$. A straightforward induction yields a bound on the distance of the perturbed trajectory from the original one: if $\xt (0)= x(0)$, then 
\begin{equation} \label{eq:dt1}
\big\| \xt (t)- x(t)\big\| \leq 
\sum_{\tau=0}^{t-1} \big\| u(t)\big\|.
\end{equation}
However, our goal is to derive stronger bounds, of the form
\begin{equation} \label{eq:dt2}
\big\| \xt (t)- x(t)\big\| \leq 
C\, \max_{k<t} \Big\|\sum_{\tau=0}^{k} u(\tau)\Big\|,
\end{equation} 
for some constant $C>0$ independent of $u(\cdot)$.
Note that a bound which is \emph{linear} in the cumulative perturbation is clearly the tightest possible, even for the trivial system $f(x)=x$.

A bound of the form \eqref{eq:dt2} is not valid in general, even for non-expansive systems, or gradient fields of convex functions \cite{AlTG18p3}.
Nevertheless, we show that such a bound is  valid for continuous-time hybrid systems driven by   a piecewise constant drift, determined by  the subdifferential 
of a piecewise linear and convex function with finitely many pieces. 
Within this class of systems, the dynamics are automatically non-expansive with respect to the Euclidean norm. 
Furthermore, this class is fairly broad, in the sense that it actually contains a seemingly larger class of non-expansive 
\emph{finite-partition hybrid systems}\footnote{In a finite-partition hybrid system, the domain is partitioned   into a finite number of regions, 
and system trajectories have a constant drift in the interior of each region.}
 \cite{AlTG18p4}.
Finite-partition systems often arise in the context of  systems that are controlled through the selection of a particular action at each time among the elements of a finite set.  
They have attracted broad interest, due to numerous applications to communication networks \cite{TassE92,Neel10}, processing systems \cite{RossBM15}, manufacturing systems, and inventory management \cite{PerkS98,Meyn08}, etc.


A prominent example to which our results apply are the fluid-level dynamics of the celebrated Max-Weight policy for real-time job scheduling \cite{TassE92}.
\hide{\footnote{In fact, our initial objective was to derive sensitivity results just for the case of Max-Weight, but 
it soon became apparent that this was a special case of a broader mathematical structure and result.
Then came to realize that it extends to a broader class of hybrid systems.}.}
This policy is used for scheduling in queueing systems: 
at each time, it chooses a service vector (from a finite set) that maximizes a weighted sum of the current queue lengths
(see Fig.~\ref{fig:application finite part is subdiff} for a simple example).
This policy and its properties, 
e.g., \emph{stability} \cite{MaguHS14,Hala14,Dai95,Tass95} 
and \emph{state space collapse} \cite{Stol04,ShahW12},
have been studied extensively over the last three 
 decades.  

When the Max-Weight policy is applied to a discrete-time stochastic setting, the perturbation  $u(\cdot)$ in \eqref{eq:dt3}
 is the sample path of a stochastic process, and captures the fluctuations in job arrivals. 
 Under  usual 
 probabilistic assumptions, $\sum_{\tau=0}^{t-1}
\|u(\tau)\|$ grows at the rate of $t$, whereas $\max_{k<t}\sum_{\tau=0}^{k}
u(\tau)$ only grows as (roughly) $\sqrt{t}$, 
with high probability.
This fact, in combination with the main result of this paper, leads to tighter  than earlier available probabilistic bounds on the fluctuations of the Max-Weight trajectories from their deterministic (fluid) counterparts, and opens  the way for new results \cite{AlTG18p2}, such as strengthening the state-space collapse results in \cite{ShahW12}.
More specifically, in \cite{AlTG18p2}, we 
study in detail the discrete-time Max-Weight dynamics: 
we
use the results of this paper to prove a bound similar to \eqref{eq:dt2}, and also address a state-space collapse conjecture  posed in \cite{ShahW12}. 
Furthermore, our approach also enables us to settle  
another open problem that was posed in  \cite{MarkMT18},  on   delay-stability in the presence of heavy-tailed traffic, as will be reported in a forthcoming paper.

\hide{
\comm{[Deleting this paragraph because it goes into a tangential topic; its relevance is not quite apparent; and for the final conclusion "beneficial to the synthesis of control methods" we do not have any examples that would back up the claim.]}
\del{Fluid models **footnote{In the simplest case, for 
a discrete time system $x^+=x+f(x)$ defined over the entire $\R^n$,  the continuous time system $\dot{x}=f(x)$ is referred to  as a corresponding fluid model.}
	have been a major tool not only for the analysis of system properties such as stability \cite{Dai95,DaiM95,Gama00}, but also for  synthesis of control algorithms. 
Optimal control decisions in stochastic networks are given by dynamic programming, which is typically computationally intractable \cite{Stid85,MaCM10}.
A popular approach around this problem is to approximate the stochastic system with its fluid model \cite{ChenDM04,BertNP15},
solve for the optimal solutions of the fluid model \cite{Meyn97,Meyn05},
and translate these solutions via discrete review methods to use them for controlling the stochastic system \cite{Magl99,Magl00,FleiS05,Meyn08}.
Here, the fluid model is a continuous time dynamical system, containing the solutions of the stochastic system as its perturbed trajectories.
Therefore, the quality of fluid model approximation,  captured in the input sensitivity bounds, is also beneficial to the synthesis of control methods.
We also wish to point out that our problem in this paper is not related to the perturbation analysis of stochastic fluid models used in optimizing the threshold parameters for optimal control of certain manufacturing systems \cite{CassWM02,YuC04,HuVY94}. 
 }
 }

As is apparent from our discussion of the Max-Weight policy, 
one may be ultimately interested in a discrete-time system, as opposed to the continuous-time systems considered in this paper. However, 
\hide{\oli{Recall that we added this paragraph because of a previous comment by reviewer 3. Now she says in her last review that this paragraph is not needed! I am shortening to save some space.}
\dela{
the two settings (discrete or continuous) are closely related. For example, in many applications involving communication networks or scheduling systems, the underlying system may evolve in discrete time, but 
much of the analysis is often carried out in terms of related continuous-time models that are easier to analyze.  This is because discrete-time models often involve uninteresting ``edge effects'' that  can result in  tedious technical details and can  obscure the essence of the underlying mathematical structure. For such reasons, }}
we found it  more natural to start with the development of the core concepts and results within the more elegant continuous-time framework in this paper, and then translate them back to the discrete-time framework. 
For instance,
\cite{AlTG18p3} shows that if a continuous-time system admits a bound of the form \eqref{eq:dt2}, then  its discrete-time counterpart obeys a similar bound. 

Regarding related literature, we are not aware of any work that resembles the main result of this paper. Some 
seemingly related research threads deal with  
\emph{input-to-state stability}\footnote{A
 discrete time dynamical system $x(t+1) = f\big(x(t),u(t)\big)$ with state $x(\cdot)$ and external disturbance (or control) $u(\cdot)$ is said to be input-to-state stable if 
there exists a continuous and strictly increasing function $\gamma:\R_+\to\R_+$ with $\gamma(0)=0$ and a function $\beta:\R^n\times\R_+\to\R_+$, strictly increasing in the first argument and decreasing  in the second argument with $\beta(0,\cdot)=0$ and $\lim_{t\to \infty}\beta(s,t)=0$, for all $s\ge 0$; such that
$\|x(t)\| \le \beta\big(\|x(0)\|  , t\big)  + \max_{k\le t}\gamma\big(\| u(k)\|\big)$, 
 for all trajectories $x(\cdot)$, all disturbances $u(\cdot)$, and all times $t$ \cite{JianW01}.
} 
\cite{JianW01,MarrAC02,Ange04,Sont08,Sont96}, 
integral input-to-state stability \cite{AngeSW00}, incremental input-to-state stability \cite{Ange02}, incrementally integral input-to-state stability \cite{Ange09}, and robust input-to-state stability \cite{CaiT13}. 
However, we note that integral input-to-state stability \cite{AngeSW00} and  incrementally integral input-to-state stability \cite{Ange09} are concerned only with generalizations of the weak bound in \eqref{eq:dt1}. Furthermore, incremental input-to-state stability 
 \cite{Ange02} involves generalizations of a sensitivity bound in terms of $b=\max_{k<t} \big\|u(\tau)\big\|$. 
A bound of the form $Cb$, for some constant $C$ that does not depend on $t$, would 
typically 
be stronger than ours; however, such a bound does not hold in our setting, even for the simplest system where $f(x)=x$. 

It is worth pointing out that for 
systems with additive disturbances, $x(t+1)=f\big(x(t)\big)+u(t)$, 
 input-to-state stability and the bound  \eqref{eq:dt2}
do not imply one another\footnote{For example, the discrete-time system  
$x(t+1)=x(t)$
   satisfies \eqref{eq:dt2},
but is not input-to-state stable.
Conversely,  the two-dimensional and two-region discrete-time system with $f(x,y)=\big(x/2,y/2\big)$ for $x\ge 0$ and $f(x,y)=\big(x/4,y/4\big)$ for $x<0$, is input-to-state stable but \eqref{eq:dt2} fails to hold.
This is because, for a trajectory initialized at $\big(x(0),y(0) \big)= (0,4)$, a small perturbation of the initial condition $\big(\tilde{x}(0),\tilde{y}(0) \big)= (-\epsilon,4)$,  will result in a distance larger than $1$ at the next time step.  
}.

Another key difference is
that input-to-state stability results usually rely on Lyapunov-type arguments
\cite{Lyap92}. 
However, Lyapunov functions seem to be inadequate for our purposes. 
This is because our bounds (as can be seen in the proof given in Section~\ref{sec:proof cont})
rely in a delicate manner on the relative orientation of the two trajectories $x(\cdot)$ and $\tilde{x}(\cdot)$, in conjunction with the local ``landscape'' of the potential function.
Furthermore, as shown in \cite{AlTG18p3}, the desired sensitivity bound \eqref{eq:dt2} fails to carry over if the number of constant-drift regions is not finite. This means that a Lyapunov-based argument would have to make essential use of our finiteness assumption, something for which we are not aware of having any precedents in the literature. 

\hide{
\comm{Is it true that what  referee 2 mentioned involves contraction mappings, in which the effect of past perturbations washes out? If so, we can add a sentence like: "Finally, works such as [REFS] rely on contraction assumptions, which wash out the effect of perturbations in the far past, something that cannot be guaranteed under non-expansiveness assumptions.''}
\oli{To my understanding, the results in the ``contraction metrics'' papers are not comparable to our bounds. See definition 3 in \href{https://www.researchgate.net/profile/Parasara_Duggirala/publication/261427359_Verification_of_annotated_models_from_executions/links/592250590f7e9b99794442ab/Verification-of-annotated-models-from-executions.pdf}{[link]} suggested by Rev 2, and also its main reference \href{https://www.researchgate.net/profile/Parasara_Duggirala/publication/261427359_Verification_of_annotated_models_from_executions/links/592250590f7e9b99794442ab/Verification-of-annotated-models-from-executions.pdf}{this one [link]}. It seems that the latter does not involve ``perturbations'', and the former provides no comparable bounds (it only argues equivalence between definitions). Moreover, the ``wash out'' effect is not evident (to me) from the definition of a contraction metric, as defined in Def3 of the former paper. 
I agree with you that there must be a wash out effect in contractive systems and also in incrementally ISS systems. However, claiming this in the paper might give rise to further controversies from the reviewer's side. Agree?}}

The rest of this paper is organized as follows. In the next section we discuss some preliminaries and our notational conventions. In Section \ref{sec:main} we
state our main theorem. In Section \ref{sec:proof cont}  we provide the core of the proof, while relegating 
 some of the details to the Appendix. Finally, 
in Section \ref{sec:discussion} we discuss possible extensions, open problems and challenges, and directions for future research.

\medskip
\section{\bf Preliminaries}\label{sec:model}
\subsection{Notation}
We denote by $\R_+$   the set of non-negative real numbers. 
For a column vector $v\in\R^n$, we denote its transpose  and  Euclidean norm  by  $v^T$ and $\Ltwo{v}$, respectively.  For  any set  $S\subseteq \R^n$,  $\Span(S)$  stands for the span of the vectors in $S$. Furthermore, if $p$ is a point in $\R^n$, then $p+S$ stands for the set $\big\{p+x\,\big|\,x\in S \big\}$, and  $d\big(p\,,\,S\big)$ for the Euclidean distance between $p$ and $S$, with the convention that  $d\big(p\,,\,S\big)=\infty$ if $S$ is empty. 
Similarly, we let $d(p,\{x\}) = \Ltwo{p-x}$ for $p,x\in \R^n$.
We finally let $A\backslash B\,=\,A\bigcap B^c$, for any two sets $A$ and $B$, where $B^c$ is the complement of $B$.

\subsection{Perturbed Dynamical Systems}
As in \cite{Stew11}, we  identify
a dynamical system with a set-valued function $F:\R^n\to 2^{\R^n}$
and the associated differential inclusion $\dot{x}(t)\in F(x(t))$. We start with a formal definition, which allows for the presence of perturbations.

\begin{definition}[Perturbed Trajectories]\label{def:integral pert traj} 
Consider a dynamical system $F:\R^n\to 2^{\R^n}$, and let $\Prt:\R\to\R^n$ be a right-continuous function, which we refer to as the \emph{perturbation}.
 Suppose that there exist measurable and integrable functions $\ptraj({\cdot})$ and $\ff({\cdot})$ of time that satisfy \begin{equation}\label{eq:def of integral pert}
\begin{split}
\ptraj(t) &= \int_0^t \ff(\tau)\,d\tau \,+\, \Prt(t),\qquad \forall\ t\ge 0,\\
\ff(t)&\in F\big(\ptraj(t)\big),\qquad \forall\ t\ge 0.
\end{split}
\end{equation}
We then call $\Prt$ the \emph{perturbation}.
Any such $\ptraj$ and $\ff$ is called a
\emph{perturbed trajectory}
and a \emph{perturbed drift},  respectively. In the special case where 
$U$ is  identically zero, we also refer to $\ptraj$  as an \emph{unperturbed trajectory}.

\end{definition}

Note that a perturbed trajectory is automatically right-continuous. In the absence of the perturbation $U(\cdot)$, Eq.~\eqref{eq:def of integral pert} becomes the differential inclusion $\dot x\in F(x(t))$ (almost everywhere). When perturbations are present,  $U$ is often absolutely continuous, of the form
$\int_0^T u(\tau)\,d\tau$, for some measurable function $u(\cdot)$. In this case, we are essentially dealing with the differential inclusion 
$\dot \ptraj(t)\in  F(\ptraj(t))+u(t)$. 
However, the integral formulation in 
Definition \ref{def:integral pert traj} is  more useful because it also
applies to cases where $U$ is not absolutely continuous, e.g., if $U$ is a sample path of a Wiener or a jump process.

\subsection{Classes of Systems}
We now introduce some classes of systems of interest. 
A dynamical system $F$ is called \emph{non-expansive} if for any pair of unperturbed trajectories $x({\cdot})$ and $y({\cdot})$, 
and if $0\leq t_1\le t_2$, then
\begin{equation}
\Ltwo{x(t_2)-y(t_{2})}\le \Ltwo{x(t_1)-y(t_1)}.
\end{equation} 

For a convex function $\Phi:\R^n\to\R^n$, we denote its subdifferential by $\partial \Phi(x)$. 
We say that $F$ is a \emph{subgradient dynamical system} if there exists a convex function $\Phi(\cdot)$, such that for any $x\in\R^n$, $F(x)=-\partial\Phi(x)$.
Furthermore, if $\Phi$ is of the form
$$\Phi(x)=\max_{i}\big(-\mu_i^Tx+b_i\big),$$ 
for some $\mu_i\in\R^n$, $b_i\in\R$, and with $i$ ranging over a {\bf finite} set,
we say that $F$ is a \emph{Finitely Piecewise Constant Subgradient} (FPCS, for short) system;
cf.~Fig.~\ref{fig:application finite part is subdiff}.

\begin{figure} 
\begin{center}
\subfigure[]{\includegraphics[width = .3\textwidth]{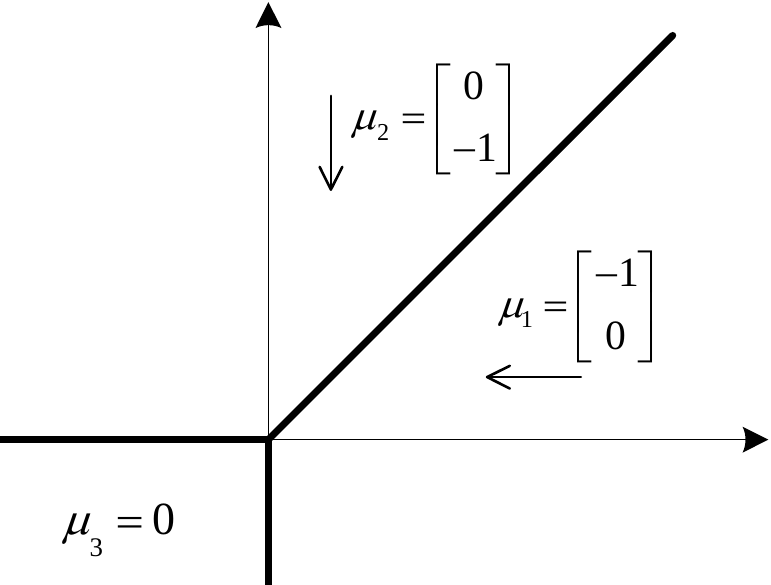}}
\subfigure[]{\includegraphics[width = .3\textwidth]{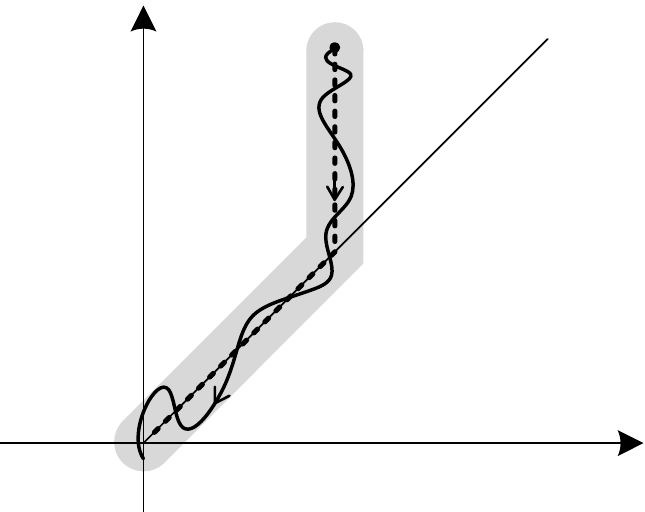}}
\ifOneCol
\vspace{-1cm}
\fi
\end{center}
\caption{
Consider a simple discrete-time network of two parallel queues with no arrivals and a unit-rate server. The Max-Weight policy always serves a longest queue.
Part (a) illustrates the corresponding continuous-time (the so called fluid-level) dynamics of this system. 
The state vector $(x_1,x_2)$ describes the workload at each queue.
(To avoid dealing with differential inclusions involving boundary constraints, we extend the naturally nonnegative state of the system to all of $\R^2$.) We have three regions indicated in the figure. The set
$F(x)$ is a singleton in the interior of each region and it is the convex hull of 
multiple vectors on the boundaries of the regions.
This dynamical system is the subgradient field of the piecewise linear convex function $\Phi(x)=\max(-\mu_1^Tx,\, -\mu_2^Tx,\, 0) =\max\{x_1,x_2,0\}$, and hence is an FPCS  system.
Part (b) depicts an unperturbed (dashed line) and a perturbed (solid line) trajectory. 
Our main result argues that the perturbed trajectory stays within a distance of the unperturbed trajectory bounded by a constant multiple of the size of the integral of the perturbation; cf. \eqref{eq:dt2}. 
}
\label{fig:application finite part is subdiff}
\end{figure}

Subgradient systems are known to have several useful properties: they are automatically non-expansive 
(cf.~Part 5 of Theorem 4.4 in \cite{Stew11}), a fact that we will be using in the sequel. 
Existence and uniqueness results are also available \cite{Stew11}. 

\begin{lemma}[Existence and Uniqueness of Solutions] \label{lem:uniqueness of trajectory of mws}
For any subgradient dynamical system $F$ and any  $x_0\in\R^n$, there exists a unique trajectory of $F$ initialized at $x_0$.
\end{lemma}
\begin{proof}
It follows from Lemma 2.30  of \cite{Stew11} that any subgradient dynamical system  is a 
maximal monotone map\footnote{A set valued function $F:\R^n\to 2^{\R^n}$ is a \emph{monotone map} if for any $x_1,x_2\in\R^n$ and any $v_1\in F(x_1)$ and $v_2\in F(x_2)$, we have $\big(v_1-v_2\big)^T\big(x_1-x_2\big)\le 0$. It is called a \emph{maximal monotone map} if it is monotone, and for any monotone map $\tilde{F}$, that satisfies $F(x)\subseteq \tilde{F}(x)$ for all $x$, we have $\tilde{F}=F$.}. 
The lemma then follows from Corollary 4.6 of \cite{Stew11}.
\end{proof}


\hide{
\green{[***REMOVING THE MATERIAL FROM HERE DOWN. DO NOT WANT TO DISTRACT THE READER TOO MUCH WITH STUFF THAT IS NOT USED.**]}

\del{The following well known  result on general subgradient dynamical systems, applies to FPCS systems, which are a special case.}

\green{NOTE: There are some issues that need fixing in this proposition.\\
1. Does the first part only apply to (\ref{eq:pert traj})?
\\Does it apply for perturbations $U(t)$? \oli{Not clear. We can think of a proof. But any proof would take much space, unless there are references. In the very first draft, before you teach me differential inclusions and give me \cite{Stew11}, there was a messy 2-page proof for existence of perturbed trajectories. I am not sure if we can modify that proof to get a similar result for integral perturbations.}\\
If not how can it be useful later?} \oli{We assume existence of $\ptraj$. Also made a small change in Theorem 1 to remove possible ambiguities.}

\green{Also, solutions to (4) may be unique, but this result allows the possibility that the integral
equation has multiple solutions, unless we strengthen the previous proposition.} \oli{It is indeed unique. See the notes after the previous proposition. However, we make no use of  uniqueness.} 

\green{Similarly, don't we want the second part to apply under Definition 1?} \oli{Seems no need to mention,  since non-expansiveness is a property of unperturbed trajectories.}

\green{
\begin{proposition}[Solvability and Non-Expansiveness of Subgradient Systems]\label{lem:existence and no expansion} 
\begin{enumerate}[label={(\alph*)}, ref={\ref{lem:existence and no expansion}(\alph*)}]
\item \label{lem:uniqueness of trajectory of mws}
For any subgradient dynamical system $F(\cdot)$, for any \dela{differential perturbation function $\pert(t)\in\measurable$} \pur{measurable and integrable function $\pert(\cdot)$}, and for any initial condition $x(0)\in\R^n$, the perturbed differential inclusion in (\ref{eq:pert traj}) has a unique solution.

\item \label{lem:non expansive}
Any subgradient dynamical system is non-expansive.
\end{enumerate}
\end{proposition}

\begin{proof}
According to Lemma 2.30  of \cite{Stew11},  any subgradient dynamical system  is a maximal monotone map\footnote{\del{A set valued function $F:\R^n\to 2^{\R^n}$ is a \emph{monotone} map if for any $x_1,x_2\in\R^n$ and any $v_1\in F(x_1)$ and $v_2\in F(x_2)$, we have $\big(v_1-v_2\big)^T\big(x_1-x_2\big)\le 0$. It is called a \emph{maximal} monotone map if it is monotone, and for any monotone map $\tilde{F}$ that satisfies $F(x)\subseteq \tilde{F}(x)$ for all $x$, we have $\tilde{F}=F$.}}. Then,  Corollary 4.6 of \cite{Stew11} implies that (\ref{eq:pert traj}) has a unique solution. The non-expansiveness is immediate from Part 5 of Theorem 4.4 in \cite{Stew11}.
\end{proof}
}
}

\hide{
\begin{remark}
Trying to say that the integral in (\ref{eq:def of integral pert}) can be any type of integral (e.g., Lebesgue, Ito?, etc.?) [the problem is that for Ito integral for example we should write $\int_0^t d\ff(\tau)$]. To circumvent such complications,  I was thinking of a formulation more general than (\ref{eq:def of integral pert}):

Given a right continuous function $\Prt:\R\to\R^n$, $\ptraj:\R\to\R^n$ is called a perturbed trajectory if
 for every $t_1\le t_2$, there exists a $\ff\in\R^n$ such that
\begin{equation}
\begin{split}
\ptraj(t_2)-\ptraj(t_1) &= \big(t_2-t_1\big)\ff \,+\, \Prt(t_2)- \Prt(t_1),\\
\ff &\in \conv\Bigg( \bigcup_{\tau\in [t_1,t_2]} F\big(\ptraj(\tau)\big)\Bigg),
\end{split}
\end{equation}
where $\conv(\cdot)$ is the convex hull of a set. The proofs also work under this definition. Does it look good?
\end{remark}
}

\medskip
\section{\bf Main Result}\label{sec:main}
We now state the  main result of the paper. 
 Its proof  is given in Section \ref{sec:proof cont}.
\begin{theorem}[Input Sensitivity of FPCS  Systems] \label{th:main cont}
Consider an FPCS   system $F$.
Then, there exists a constant $\divconst$ such that for any unperturbed trajectory $x(\cdot)$, and for any 
perturbed trajectory $\ptraj(\cdot)$ with corresponding  perturbation $\Prt(\cdot)$ and the same initial conditions $\ptraj(0)=x(0)$, 
\begin{equation} \label{eq:bounded pert cont}
\Ltwo{\ptraj(t)- x(t)} \,\le \, 
\divconst\, \sup_{\tau\le t} \Ltwo{ \Prt(\tau)}, \qquad \forall\ t\in \R_+.
\end{equation}
Moreover, for any $\lambda\in\R^n$, the bound \eqref{eq:bounded pert cont} applies to the 
(necessarily FPCS)   system 
$F(\cdot)+\lambda$ 
 with the same constant~$\divconst$.
\end{theorem}
 \medskip

%

Theorem \ref{th:main cont} is limited to FPCS systems:  if any of the assumptions in the definition of FPCS  systems is removed, then a similar result is no longer possible. In a forthcoming document, 
we  discuss several examples of dynamical systems for which 
no constant $\divconst$ satisfies
(\ref{eq:bounded pert cont}); cf.~Section \ref{sec:discussion}. 

We finally note that the vector $\lambda$ in the dynamical system $ F(\cdot)+\lambda$ can be viewed as a constant external field. Thus, the second part of the theorem  asserts that the same bound holds uniformly for all constant external fields.

The proof of Theorem 1, presented in the next section, is fairly involved and so it is useful to provide some perspective on the challenges that are involved. For a constant-drift system, of the form $\dot{x}(t) = \mu+u(t)$, the result is immediate, because the state is fully determined by the integral $\int_0^t u(\tau)\,d\tau$. More generally, the unperturbed system goes through successive constant-drift regions, and one might expect that the result can be obtained by deriving and patching together bounds for each region encountered. There is however a difficulty, because the unperturbed trajectory often lies at the intersection of the boundaries of two or more constant drift regions. When that happens, the perturbed trajectory may chatter between different regions. As a consequence, the number of pieces and bounds that would have to be patched together can become arbitrarily large, and a bound of the desired form does not follow. For this reason, we need a much more refined analysis of the trajectories in the vicinity of the intersection of different regions, as will be seen in the next section.

\medskip
\section{\bf Proof}\label{sec:proof cont}
In this section we present the proof of Theorem \ref{th:main cont}, organized in a sequence of three subsections. In Subsection \ref{subsec:properties of unperturbed} we present some notation, definitions, and lemmas, mostly concerning the geometric properties of FPCS  systems and unperturbed trajectories. In particular, we define critical points (Definition \ref{def:cp}) as the extreme points of constant-drift regions.

In  Subsection \ref{sub:b} we consider 
a time interval during which
 the perturbed trajectory is far from the set of critical points. 
 Such an interval can be divided into subintervals with an important property: the set of drifts encountered is low-dimensional, in a sense to be defined below. 
Within each such subinterval,
we show in Lemma \ref{lem:locally close} that  the local dynamics are equivalent to the dynamics of a lower-dimensional FPCS system, and employ a suitable induction on the system dimension to obtain a certain upper bound. Then, in Proposition \ref{prop:close to fluids at far}, we piece together the bounds for the different subintervals to obtain an upper bound that applies as long as the perturbed trajectory remains far from the set of critical points. 
 
In Subsection \ref{sub:c} we consider the case where the perturbed trajectory comes close to a critical point: 
we show, in
 Proposition \ref{prop:close in the basin},
  that the unperturbed trajectory stays close to the perturbed trajectory, as long as the perturbed trajectory remains sufficiently close to that critical point.  Finally, 
 in Subsection  \ref{sub:d} we combine the two cases and bound the distance of the trajectories at all times.

From now on, 
we assume that $x(0)=\ptraj(0)$ and that
\begin{equation} \label{eq:bounded pert at all times}
\sup_{t} \Ltwo{\Prt(t)} \le \maxpert.
\end{equation}
We will show that for any  $t\ge0$, we have $\Ltwo{\ptraj(t) - x(t)}\le \divconst \maxpert$, for some constant $\divconst$ independent of $\Prt$, $\theta$, and $x(0)$. It is not hard to see that this implies the theorem in its original form.

The proof proceeds by induction on the system dimension $n$.
In particular, we make the following {\bf induction hypothesis}, which we assume to be in effect  throughout the rest of this section.
\ifOneCol
\begin{equation} \label{induc:induction on dimension}
\textrm{Induction hypothesis:} \qquad \begin{array}{l}
\textrm{Theorem \ref{th:main cont} holds for all $(n-1)$-dimensional} \\ \textrm{FPCS  systems.}
\end{array} 
\end{equation}
\else
\begin{equation} \label{induc:induction on dimension}
\begin{array}{l} \textrm{Induction} \\ \textrm{hypothesis} \end{array} \textrm{\bf:} \quad \begin{array}{l}
\textrm{Theorem \ref{th:main cont} holds for all} \\ \textrm{$(n-1)$-dimensional FPCS  systems.}
\end{array} 
\end{equation}
\fi

We then rely on the induction hypothesis to  prove the theorem for $n$-dimensional systems. 
For the basis of the induction we 
consider
the case of zero-dimensional systems. In this case, the state space consists of a single point (the zero vector),  we have $\x(t)=\xt(t)=0$ at all times, and the result in 
Theorem \ref{th:main cont} holds trivially.

\subsection{Properties of Unperturbed Dynamics}\label{subsec:properties of unperturbed}
In this subsection we present some notation and definitions, and prove some properties of unperturbed trajectories.  We then define and study critical points. 
Throughout the proof, we assume that $F$ is an FPCS system on $\R^n$, with $F=-\partial \Phi$, where 
$\Phi(x)=\max_{i=1,\ldots,m} \big(-\mu_i^T x+b_i\big)$. 
 We assume that the vectors $\mu_i$ in the definition of $\Phi$ are distinct. This entails no loss of generality, because if $\mu_i=\mu_j$ and $b_i>b_j$, then $-\mu_j^T x+b_j$ is always dominated by $-\mu_i^T x+b_i$ and has no effect on $\Phi(\cdot)$.  

Each vector $\mu_i$ is called a \emph{drift} and 
we define $\actionset$ to be the set $\big\{  \mu_i \big\}_{i=1}^m$  of all drifts. 
For each drift $\mu\in\actionset$, we use the notation $b_\mu$ to refer to the corresponding constant in the expression for $\Phi$.  With these conventions, we have
\begin{equation}
\Phi(x)=\max_{\mu\in\actionset} \big(-\mu^T x+b_\mu\big).
\end{equation}

For every $x\in\R^n$, we define the set 
of active drifts at $x$ as 
\begin{equation}
\maxactset(x)\triangleq \left\{ \mu\in\actionset \,\big|\,\Phi(x)=-\mu^Tx+ b_\mu  \right\}.
\end{equation}
If at some $x$ the corresponding set $\maxactset(x)$ consists of a single element $\mu$, we have $\dot  x={\mu}$. 
However, the dynamics become more interesting when 
$\maxactset(x)$ contains multiple elements. For that case, 
it follows from the definition of the subdifferential that for any $x\in\R^n$, 
$F(x)$ is the convex hull of $\maxactset(x)$. 

For each $\mu \in \actionset$, we define its \emph{effective region} $\region_\mu$ by
\begin{equation}
\region_\mu \triangleq \left\{x\in\R^n \,\big|\,   \mu\in \maxactset(x)  \right\}.
\end{equation}
Equivalently, 
$$\region_\mu =\left\{ x \,\big|\, -\mu^Tx+ b_{\mu} \geq -\nu^T x+b_{\nu},\ \forall\ \nu\in \actionset\right\},$$
which establishes that each 
 region $\region_\mu$ is a polyhedron and, in particular, closed and convex.  
We will be using $\regionset$ to denote the collection of
all effective regions: $\regionset\triangleq\big\{ \region_\mu\,\big|\, \mu\in\actionset \big\} $.

From now on, and with some abuse of traditional notation, we will use $\dot{x}(t)$ to denote the {\bf right-derivative} of $x(t)$, whenever it exists.
The lemma that follows shows that for unperturbed trajectories this right derivative always exists and has some remarkable properties.

\begin{lemma}[Properties of Unperturbed Trajectories] \label{lem:unperturbed norm and decreasing drift}
Let $x(\cdot)$ be an unperturbed trajectory of an FPCS  system $F$. Then,
\begin{enumerate}[label={(\alph*)}, ref={\ref{lem:unperturbed norm and decreasing drift}(\alph*)}]
\item(Minimum Norm)\quad \label{lem:min norm d+}
For every  $t\ge0$, the right derivative of $x(t)$  exists and is given by
\begin{equation}\label{eq:d+ is unique min of Ltwo}
 \dot\x(t) = \argmin{v\in F(\x(t))} \Ltwo{v},
\end{equation}
with the minimizer being unique. 

\item(Decreasing Drift Size)\quad  \label{lem:decreasing drift size}
If $t>s$, then $\Ltwo{\dot\x(t)}\le\Ltwo{\dot\x(s)}$, and the inequality is strict if ${\dot\x(t)}\ne{\dot\x(s)}$. Furthermore, an unperturbed trajectory traverses a connected sequence of at most $2^m-2$ line segments, possibly followed by a half-line.
\end{enumerate}
\end{lemma}

\begin{proof}
The first part of the lemma is an immediate consequence of  Part 3 of Theorem 4.4 in \cite{Stew11}. 
For Part (b), we invoke 
Part 4 of Theorem 4.4 in \cite{Stew11} which states that $\Ltwo{\dot\x(t)}$ is a non-increasing function of time. Since for any $x$, $F(x)$ is the convex hull of $\maxactset(x)$,  it follows from (\ref{eq:d+ is unique min of Ltwo}) that  $\dot{x}(t)$ is uniquely determined by $\maxactset\big(x(t)\big)$. There are at most $2^m-1$ non-empty subsets $\maxactset(x)$ of $\actionset$. Hence, $\dot\x(t)$ can take at most $2^m-1$ different values. 

Fix a time $s\ge0$ and let $t$ be the infimum of the times $\tau>s$ for which $\dot\x(\tau)\ne \dot\x(s)$. 
The time function $\dot\x(\cdot)$ is piecewise constant and right-continuous (Part 4 of Theorem 4.4 in \cite{Stew11}). This implies that $t>s$ and 
$\dot\x(t)\ne \dot\x(s)$. Furthermore, from the strict convexity of the Euclidean norm we obtain
\begin{equation}\label{eq:middle point has less norm}
\Ltwo{\big(\dot\x(s)+\dot\x(t)\big)/2}< \max{}\big(\Ltwo{\dot\x(s)}\,,\, \Ltwo{\dot\x(t)}\big).
\end{equation} 
Since every region $\region_\mu$ is closed, there exists a sufficiently small neighbourhood $\ball$ of $x(t)$ such that if $x(t)\not\in\region_\mu$, then $\ball$ does not intersect $\region_\mu$.  
Equivalently, 
for any $y\in\ball$, we have $\maxactset(y)\subseteq\maxactset\big(x(t)\big)$, and $F(y) \subseteq F\big(x(t)\big)$. In particular, 
 consider a $\tau\in\big[s,t\big)$ such that $x(\tau)\in\ball$. Then, $\dot{x}(s)=\dot{x}(\tau)\in F\big(x(\tau)\big)\subseteq F\big(x(t)\big)$. 
 Since $F\big(x(t)\big)$ is convex, $\big(\dot\x(s)+\dot\x(t)\big)/2\in F\big(\x(t)\big)$.  Therefore,  (\ref{eq:d+ is unique min of Ltwo}) implies that $\Ltwo{\dot\x(t)} \le \Ltwo{\big(\dot\x(s)+\dot\x(t)\big)/2}$. 
 Together with (\ref{eq:middle point has less norm}), this shows that $\Ltwo{\dot\x(t)}<\Ltwo{\dot\x(s)}$. 
 
For the last statement in Part (b) of the lemma, note that there are 
at most $2^m-1$ number of different possible sets  $\maxactset(x)$,
and therefore as many choices for $F(x)$. Using \eqref{eq:d+ is unique min of Ltwo}, there are at most $2^m-1$ possible values for $\dot\x(t)$. 
As we have already shown that  $\Ltwo{\dot\x(t)}$ decreases strictly each time that it changes,  an unperturbed trajectory consists of at most $2^m-1$ pieces, with a constant derivative on each piece.  
This implies that the trajectory traverses a connected sequence of at most $2^m-2$ line segments, possibly followed by a half-line. 
\end{proof}

For any $x\in\R^n$, consider the unperturbed trajectory $z(\cdot)$ initialized with $z(0)=x$. We define the \emph{actual drift} at 
$x$ as 
$\dir(x)\triangleq \dot z(0)$,
where we continue using the convention that $\dot z$ stands for the right derivative.
According to Lemma \ref{lem:min norm d+}, the actual drift always exists and is uniquely determined by $x$.

We now proceed to define \emph{critical points}, 
which will play a central role in the sequel. 

\begin{definition}[Critical Points]\label{def:cp}
A point $p\in\R^n$ is called a \emph{critical point} if $\mathrm{span}\big(\big\{ \mu-\mu'  \,\big|\, \mu,\mu' \in\maxactset(p)  \big\}\big)= \R^n$. The set of critical points is denoted by $\CP$.
\end{definition}
An equivalent condition is that for a critical point $p$, the affine span of $\maxactset(p)$, i.e., the smallest affine space that contains $\maxactset(p)$, is equal to the entire set $\R^n$. For this to happen, $\maxactset(p)$ must have at least $n+1$ elements, and therefore $p$ must lie at the intersection of at least $n+1$ regions (although the converse is not always true). For the example in Fig.~\ref{fig:application finite part is subdiff}, $p=0$ is the only candidate and is in fact a critical point because the affine span condition is satisfied.
Furthermore, it will be shown in Lemma \ref{lem-cp: finite number} that the critical points are the extreme points of the regions $\region_{\mu}$. 

\begin{definition}[Basin of a Critical Point] \label{def:basin of a cp}
Consider some  $\basin\in\R_+\cup\{\infty\}$ and a critical point $p$, with actual drift $\dir(p)$ equal to $\dir$.
 The closed ball $\ball$ of radius $\basin$ centered at $p$ is called a basin of $p$ (and $\basin$ is called a basin radius for $p$)  if for every $x\in \ball$ and every $y\in F(x)$, we have $\dir^T y\ge \lVert \dir \rVert^2$.
\end{definition}
Note that the inequality $\dir^T y\ge \lVert \dir \rVert^2$ implies that $\| \xi\| \leq \| y\|$. 
As a result, $\dir(p)$ has the minimum norm among all possible drifts within the basin of a critical point $p$. Also note that basins of a critical point $p$ are not necessarily unique: if the radius $\basin$ is positive, another basin is obtained by reducing the radius. Figure \ref{fig:critical points} shows an example of a 
two-dimensional  system with three critical points and some associated basins. Basins of critical points will appear frequently in the sequel. 

Before moving to study the properties of critical points, we introduce one last definition.

\begin{figure} 
\begin{center}
{
\includegraphics[width = .25\textwidth]{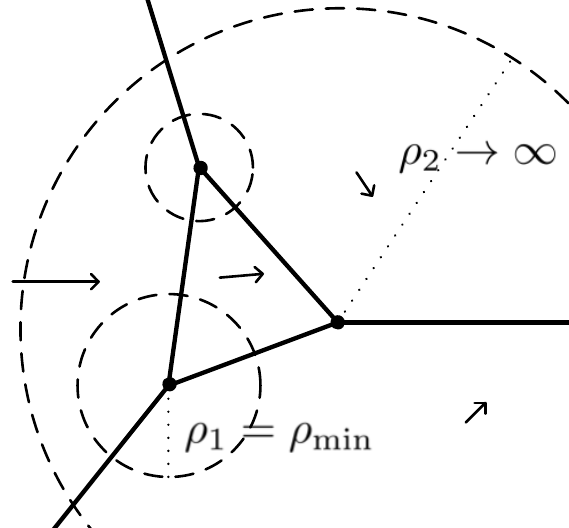}}
\end{center}
\caption{A two-dimensional FPCS system with four regions and three critical points. The balls around the critical points show examples of associated basins. Here, the basin of the rightmost critical point can be taken equal to $\R^n$, thus containing the basins of all other critical points; cf.~Lemma~\ref{lem:cp}(c).
The CNC $\basinmin$ (cf.~Definition \ref{def:def basinmin}) is also shown.
}
\label{fig:critical points}
\end{figure}

\begin{definition}[Conic Neighbourhood Constant] \label{def:def basinmin}
We define the \emph{Conic Neighbourhood Constant} (CNC), denoted by $\basinmin$, as 
\begin{equation}\label{eq:def basinmin}
\basinmin \triangleq \frac{1}{2} 
\min \left\{d\big(p,\region \big) \,\Big|\, p\in\CP,\, \region\in\regionset,\, p\notin\region   \right\}, 
\end{equation}
i.e., $\basinmin$ is half of the minimum over all critical points, of the distance of a critical point from the regions that do not contain it. We use the convention that the minimum of an empty set is infinite.
\end{definition}
Note that  $\basinmin$ is always positive (and possibly infinite). 
We say that a dynamical system $F$  is \emph{conic} if $F=-\partial \Phi$, where $\Phi$ is  of the form $\Phi(x)=\max_i \big\{ -\mu_i^T\big(x-p\big)\big\}$, for some $p\in\R^n$. It is not hard to see that for such a conic system, either  $p$ is the only critical point or no critical points exist. 
It turns out that the ``local'' dynamics in the CNC-neighbourhood of a critical point of a general system are conic, hence the name CNC. 

The   lemma that follows lists a number of useful properties of critical points.

\begin{lemma}[Properties of Critical Points] \label{lem:cp}
Consider an FPCS  system $F$, with an associated set of critical points $\CP$. 

\begin{enumerate}[label={(\alph*)}, ref={\ref{lem:cp}(\alph*)}]
\item \label{lem-cp: finite number}
A point in a region $\region_\nu$ is a critical point if and only if it 
is an extreme point of  $\region_\nu$.
In particular, 
there are finitely many critical points. \medskip
\item \label{lem-cp: basin straight movement}
Consider a critical point $p\in\CP$ and a basin radius $\basin$ for $p$. 
Let $z({\cdot})$ be the unperturbed trajectory with initial point $z(0)=p$, and let $\dir=\dot{z}(0)$ be the actual drift at $p$. Then, 
before the time that $z(\cdot)$ exits the basin, 
$\dot{z}(t)$ is constant,
 and $z(t)=p+t\dir$, for all $t\in\big[0,\,{\basin}/{\lVert \dir\rVert}\big]$.
\hide{\green{[The word "until" could make one think that it stays constant until then, and then changes. Hence the rewording.]}}

\medskip
\item \label{lem-cp: basin infinite rad}
If $\CP$ is non-empty, then there exists a critical point $p\in\CP$ such that the entire set $\R^n$ is a basin of $p$. In the special case where $F$ is conic with a unique critical point $p$, the entire set $\R^n$ is a basin of $p$.

\medskip
\item \label{lem-cp: basin cnc}
The CNC, $\basinmin$, defined in (\ref{eq:def basinmin}), is a basin radius for every critical point. 

\medskip
\item \label{lem-cp: no revisit}
Consider a basin radius $\basin$ of a critical point $p\in\CP$,   an unperturbed trajectory $x(\cdot)$, and times $t_1<t_2$. Suppose that $\Ltwo{x(t_1)-p}\le{\basin}/{3}$ and $\Ltwo{x(t_2)-p}>\basin$. Then, for any $t\ge t_2$, $\Ltwo{x(t)-p}>{\basin}/{3}$.
\hide{
\medskip
\item \label{lem-cp: close if ball bad}
For the CNC, $\gamma$, and for any $r\ge0$ and $x\in\R^n$ for which
\begin{equation}\label{eq:dim span FBrx less than n}
\mathrm{dim }\bigg( \mathrm{span }\bigg\{ \bigcup_{y\in B_r(x)}  F(y)    \bigg\} \bigg)<n,
\end{equation}
we have $d\big(x,\CP\big)\le \gamma r$, where in the above equation $B_r(x)$ is the closed $r$-neighbourhood of $x$.
}

\medskip
\item \label{lem-cp: F' vs F}
Fix some $\lambda\in\R^n$ and consider   $F'(\cdot)\triangleq F(\cdot)+\lambda$, which is also an FPCS system. Then, 
$F$ and $F'$ have
the same set of regions  $\regionset$,  the same set of critical points  $\CP$, and the same CNC  $\basinmin$. 

\end{enumerate}
\end{lemma}
In words, part (e) states that an unperturbed trajectory that starts near a critical point $p$ and later goes sufficiently far from $p$, will never come back close to $p$.
The proof of Lemma \ref{lem:cp} is given in Appendix \ref{app:proof cp}.

\medskip
\subsection{Bounding the Deviation when the Trajectories are Far from the Set of Critical Points}
\label{sub:b}
In this subsection we bound the distance between perturbed and unperturbed trajectories,  
for the case where the perturbed trajectory stays
far from the set of critical points. To do this, we will show that when far from the set of critical points, the local dynamics are similar to those of a 
lower-dimensional system, and then use the induction hypothesis (\ref{induc:induction on dimension}).

We start with some definitions.
For any $x\in\R^n$ and $r>0$, let
\begin{equation}\label{eq:prop delta def U}
\ldset_r(x) \triangleq \bigcup_{y:\, \lVert{y-x}\rVert_2\le r}\, \maxactset(y),
\end{equation}
which is the set of  possible 
drifts in the $r$-neighbourhood of~$x$. 

\begin{definition}[Low-Dimensional Sets] \label{def:ess low dim}
We call a subset $\ldset \subseteq \R^n$ \emph{low-dimensional} if $\mathrm{span }\left\{ x-y\,\big|\, x,y\in \ldset    \right\} \ne \R^n$. 
\end{definition}
Equivalently, a set is low-dimensional if its affine span is not the entire space.  
If $x$ is a critical point, then, by definition, the vectors in $\left\{\mu_i-\mu_j \mid \mu_i,\mu_j\in \maxactset(x)\right\}$ span $\R^n$ and the set $\ldset_r(x)$ is \emph{not} low-dimensional, for any $r>0$. On the other hand, as asserted by the next lemma, which is
proved in Appendix \ref{app:proof gamma} (available in supplementary materials),
$\ldset_r(x)$ is 
low-dimensional when $x$ is sufficiently far from critical points.

\begin{lemma}\label{lem:distance from cp}
Consider an FPCS system with an associated set of critical points $\CP$. There exists $\gamma\ge 1$ 
 such that if $r>0$ and  $d\big(x,\CP\big)>\oneoverdel r$, then $\ldset_r(x)$ is low-dimensional.
\end{lemma}


In the sequel, it will be convenient to compare the perturbed trajectory with an unperturbed trajectory that starts at the same state at some intermediate time. This motivates the following terminology.
\begin{definition}[Coupled Trajectories]
Let  $\maxpert,\,T\geq 0$  be some constants. 
Let $x(\cdot)$ be an unperturbed trajectory. Let $\ptraj(\cdot)$ be a perturbed trajectory with a perturbation $\Prt(\cdot)$ that satisfies 
 (\ref{eq:bounded pert at all times}). If in addition we have $\ptraj(T)=x(T)$, we then say that 
$x(\cdot)$ and $\ptraj(\cdot)$ are
\emph{ $\maxpert$-coupled at time}
 $T$. 
 \end{definition}

\begin{figure*} 
\begin{center}

\subfigure[]{\includegraphics[width = .4\textwidth]{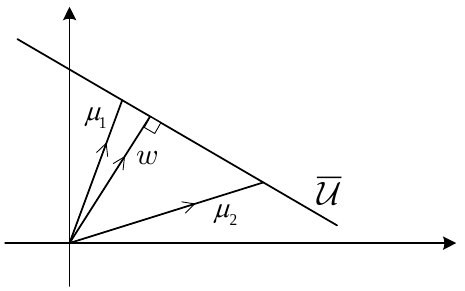}}\qquad
\subfigure[]{\includegraphics[width = .35\textwidth]{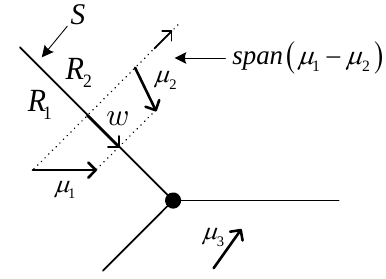}}
\vspace{.3cm}
\subfigure[]{\includegraphics[width = .8\textwidth]{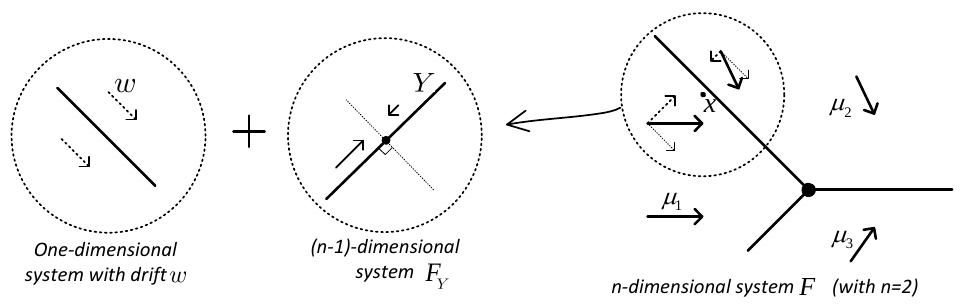}}
\vspace{-.5cm}
\end{center}
\caption{(a) Illustration of the affine space $\bar{\ldset}$ and its minimum norm element $w$,  when $\ldset$ consists of two elements $\mu_1$ and $\mu_2$.  (b) Consider the intersection $S=\{x \mid -\mu_1^Tx+b_{\mu_1}= -\mu_2^Tx+b_{\mu_2}\}$ of the two regions $\region_{\mu_1}$ and $\region_{\mu_2}$, and note that $S$ is orthogonal to the span of $\mu_1-\mu_2$. 
In this case, $w$ is  a direction of motion along $S$. (c) The dynamics can be decomposed as the superposition of a movement along $S$, in the direction of  $w$,  together with 
lower-dimensional hybrid dynamics in directions orthogonal to $S$.
}
\label{fig:2d example local projection}
\end{figure*}

The  proof will now continue along the  following lines. When far enough from the set of critical points, the set $\ldset_r(x)$ is  low-dimensional 
(Lemma \ref{lem:distance from cp}). This yields a description  of 
the dynamics  
as the superposition of an essentially $(n-1)$-dimensional FPCS system  and a constant drift.


\begin{lemma}\label{lem:locally close}
Consider an FPCS system. There exists a constant $\sigma\ge 1$ such that the following statement holds for all $T,\maxpert > 0$.  
Let $x({\cdot})$ and $\ptraj({\cdot})$ be a pair of $\maxpert$-coupled trajectories at time $0$. 
Suppose that  $\ldset\subseteq \actionset$ is low-dimensional,  and that $\ldset_{ \sigma \maxpert}\big(x(t)\big)\subseteq \ldset$, for all $t\in[0,T]$. Then, 
\begin{equation}\label{eq:desired bound!}
\Ltwo{\ptraj(t)-x(t)}\le \sigma\maxpert,\qquad \forall \ t\in[0,T].
\end{equation}
Moreover, for any $\lambda\in\R^n$, the same constant $\sigma$ also applies to the FPCS system $F(\cdot)+\lambda$.
\end{lemma}

\begin{proof} 
Let us fix some $\mu\in\ldset$. Let $\bar{\ldset}$ be the affine span of $\ldset$:
\begin{equation}
\bar{\ldset} \triangleq \mu +\Span\big\{\nu-\mu\,\big|\, \nu\in\ldset\big\}.
\end{equation}
Note that any choice of $\mu\in\ldset$ leads to the same set $\bar{\ldset}$. Let $w$ be the projection of $0$ onto $\bar\ldset$, i.e., the smallest norm element of $\bar\ldset$; cf.~Fig.~\ref{fig:2d example local projection}(a). 
Since $w\in \bar{\ldset}$, we have 
\begin{equation}\label{eq:mu-w in U}
w-\mu\in \Span\big\{\nu-\mu\,\big|\, \nu\in\ldset\big\}.
\end{equation}
Furthermore, by the orthogonality 
of projections,  $w$  is orthogonal to the difference of any two elements of $\bar{\ldset}$. In particular, 
\begin{equation}\label{eq:w perp differences}
\begin{split}
w^T\big( \mu-\nu \big)&=0, \qquad \forall\ \nu\in \bar{\ldset},\\
w^T\big( \mu -w \big)&=0.
\end{split}
\end{equation}

Since $\ldset$ is low-dimensional, $\Span\big\{\nu-\mu\,\big|\, \nu\in\ldset\big\}$ is a proper subset of $\R^n$. 
Let $Y$ be a subspace of dimension $n-1$ that contains $\Span\big\{\nu-\mu\,\big|\, \nu\in\ldset\big\}$ and is orthogonal to  $w$.\footnote{If $w\ne 0$, then $Y$ must be the orthogonal complement of the one-dimensional space spanned by $w$. If $w=0$, then
any $(n-1)$-dimensional subspace that contains $\Span\big\{\nu-\mu\,\big|\, \nu\in\ldset\big\}$ will do, 
and the choice of $Y$ need not be unique.}
Note that by the definition of $\bar{\ldset}$,
\ifOneCol
\begin{equation}\label{eq:U - w in Y}
\bar{\ldset} -w = \mu-w +\Span\big\{\nu-\mu\,\big|\, \nu\in\ldset\big\} = \Span\big\{\nu-\mu\,\big|\, \nu\in\ldset\big\} \subseteq Y,
\end{equation}
\else
\begin{equation}\label{eq:U - w in Y}
\begin{split}
\bar{\ldset} -w &= \mu-w +\Span\big\{\nu-\mu\,\big|\, \nu\in\ldset\big\} \\
&= \Span\big\{\nu-\mu\,\big|\, \nu\in\ldset\big\} \\
&\subseteq Y,
\end{split}
\end{equation}
\fi
where the second equality is due to (\ref{eq:mu-w in U}).

Any vector has an orthogonal decomposition as the sum of its projections on $Y^\perp$ (the orthogonal complement of $Y$) and $Y$; we use the subscripts $w$ and $Y$ to indicate the corresponding components, e.g.,
\begin{equation}
\begin{split}
x &= x_w+x_Y,\\
\ptraj &= \ptraj_w+\ptraj_Y,\\
\Prt &= \Prt_w +\Prt_Y.
\end{split}
\end{equation}
We will now show that the $w$ and $Y$ components of a trajectory evolve without interacting, according to a one-dimensional system with drift $w$, and an $(n-1)$-dimensional system $F_Y$, respectively; see Fig.~\ref{fig:2d example local projection}(c)  for an illustration. 


\begin{claim}\label{claim:FY}
Consider some $x\in\R^n$, and suppose that $\actionset(x)\subseteq \ldset$. 
Then, $F(x)=w+F_Y(x_Y)$, where $F_Y: Y\to 2^Y$  is an FPCS system on the $(n-1)$-dimensional subspace $Y$.
\end{claim}
\begin{proof}[Proof of Claim]
Let $\bar{F}_Y(x) = F(x)-w$.
Since $F(x)$ is contained in the convex hull of $\ldset$,  it is also in the affine span of $\ldset$, i.e., $F(x)\subseteq\bar{\ldset}$.   Then, (\ref{eq:U - w in Y}) implies that
\begin{equation}
\bar{F}_Y(x) = F(x)-w \subseteq \bar{\ldset}-w \subseteq Y.
\end{equation}
On the other hand, $\bar{F}_Y(x)$, being equal to $F(x)-w$, is the negative of the subdifferential of
\begin{equation}\label{eq:def phi U}
\bar\Phi_Y(x) \triangleq  \max_{\mu\in \ldset} \Big[\big(-(\mu-\trend\big)^T x + b_\mu\Big];
\end{equation}
we 
used here the assumption $\actionset(x)\subseteq \ldset$.
For any $x\in\R^n$,
\begin{equation}\label{eq:phi bar = phi}
\begin{split}
\bar\Phi_Y(x) &= \bar\Phi_Y(x_w+x_Y) \\
&=\max_{\mu\in \ldset} \Big[-\big(\mu-\trend\big)^T \big(x_w+x_Y \big) + b_\mu\Big]\\
&= \max_{\mu\in \ldset} \Big[-\big(\mu-\trend\big)^T x_Y + b_\mu\Big] \\
&\triangleq \Phi_Y(x_Y),
\end{split}
\end{equation}
where the third equality is due to (\ref{eq:w perp differences}). Let $F_Y:Y\to2^Y$
 be equal to $-\Phi_Y({\cdot})$. Then, $F_Y$ is 	 an FPCS  system on the $(n-1)$-dimensional subspace $Y$.
It follows from (\ref{eq:phi bar = phi}) that $\bar\Phi_Y(x)$ only depends on $x_Y$. Therefore, its negative subdifferential $\bar{F}_Y(x)$ also only depends on $x_Y$, and $\bar{F}_Y(x) = F_Y(x_Y)$. Hence, the definition $\bar{F}_Y(x) = F(x)-w$ implies that $ F(x) = w + \bar{F}_Y(x) = w + F_Y(x_Y)$, which establishes the claim.
\end{proof}

We now appeal to the induction hypothesis (\ref{induc:induction on dimension}), 
and let $\divconst_\ldset$ be equal to the constant $\divconst_Y$
 of Theorem \ref{th:main cont} for the $(n-1)$-dimensional FPCS system $F_Y$. 
Let $\sigma$ be the maximum of all such constants $\divconst_{\tilde{\ldset}}$ plus 4, over all low-dimensional subsets $\tilde{\ldset}\subseteq\actionset$:
\begin{equation}\label{eq:def sigma}
\sigma \triangleq \max\left\{ \divconst_{\tilde{\ldset}} \,\big|\,\, {\tilde{\ldset}}\subseteq\actionset \,\textrm{ and }\, \tilde{\ldset} \textrm{ is  low-dimensional}    \right\} +4.
\end{equation}
Suppose now that we add 
a constant drift $\lambda$ to $F$.
We observe that for any given low-dimensional $\ldset$, the resulting set-valued mapping
$F_{Y}$ only changes through the addition of a constant drift $\lambda_Y$; its structure remains otherwise the same.  Hence, according to the induction hypothesis (\ref{induc:induction on dimension}), $\divconst_\ldset$ is not affected when we add a constant drift $\lambda\in\R^n$ to the dynamics. 
As a consequence, the value of $\sigma$ associated with a system $F(\cdot)$ remains the same when we consider the system 
$F(\cdot)+\lambda$.

We now return to the main part of the proof of the lemma. We argue by contradiction, and
assume that (\ref{eq:desired bound!}) fails to hold. Then, from the right-continuity of $x(t)$ and $\ptraj(t)$, there exists a time $\tilde{T}\le T$ such that
\begin{equation}\label{eq:at Tr the distance iis equal to t}
\begin{split}
&\Ltwo{\ptraj\big( \tilde{T}\big)-x\big(\tilde{T}\big)}\ge \sigma\maxpert,\\
&\Ltwo{\ptraj\big( t\big)-x\big(t\big)} <\sigma\maxpert , \quad \forall\ t<\tilde{T}.
\end{split}
\end{equation} 
It follows from (\ref{eq:at Tr the distance iis equal to t}) and the assumption $\ldset_{\sigma\maxpert}\big(x(t)\big)\subseteq \ldset$  that,  for any $t<\tilde{T}$, $\maxactset\big(x(t)\big) = \ldset_0\big(x(t)\big)\subseteq \ldset_{\sigma\maxpert}\big(x(t)\big) \subseteq\ldset$.
Furthermore,
 $\maxactset\big(\ptraj(t)\big) \subseteq \ldset_{\sigma\theta}\big(x(t)\big)\subseteq  \ldset$.

Consider some $t<\tilde{T}$ and let $\tilde{\ff}(t)\in F\big( \ptraj(t) \big)$ be a perturbed drift associated with the perturbed trajectory $\ptraj(\cdot)$ (cf. Definition \ref{def:integral pert traj}).
It follows from Claim \ref{claim:FY} that $\tilde{\ff}(t)-w\in F\big( \ptraj(t) \big)-w=F_Y\big( \ptraj(t) \big)
\subseteq Y$. 
Thus, the orthogonal decomposition of $\tilde{\ff}(t)$ yields $\tilde{\ff}_Y(t)=\tilde{\ff}(t)-w$. This allows us to develop an orthogonal decomposition of $\ptraj(t)$, as follows: 
\ifOneCol
\begin{equation}\label{eq:decompose xtilde into orthogonal directions 1} 
\begin{split}
\ptraj(t) &=  \ptraj(0) + \int_0^t \tilde{\ff}(\tau)\,d\tau   + \Prt(t)\\
&= \ptraj(0) + \int_0^t \big(w+ \tilde{\ff}_Y(\tau)\big)\,d\tau   + \Prt(t)\\
&= \Big[\ptraj_w(0) + wt   + \Prt_w(t)\Big]+ \left[\ptraj_Y(0) + \int_0^t \tilde{\ff}_Y(\tau)\,d\tau   + \Prt_Y(t)\right]\\
& = \ptraj_w(t) + \ptraj_Y(t), 
\end{split}
\end{equation}
\else
\begin{equation}\label{eq:decompose xtilde into orthogonal directions 1} 
\begin{split}
\ptraj(t) &=  \ptraj(0) + \int_0^t \tilde{\ff}(\tau)\,d\tau   + \Prt(t)\\
&= \ptraj(0) + \int_0^t \big(w+ \tilde{\ff}_Y(\tau)\big)\,d\tau   + \Prt(t)\\
&= \Big[\ptraj_w(0) + wt   + \Prt_w(t)\Big]\\
&\quad\,+ \left[\ptraj_Y(0) + \int_0^t \tilde{\ff}_Y(\tau)\,d\tau   + \Prt_Y(t)\right]\\
& = \ptraj_w(t) + \ptraj_Y(t), 
\end{split}
\end{equation}
\fi
where the last equality follows because the two terms inside brackets belong to $Y^{\perp}$ and $Y$, respectively, and therefore provide the orthogonal decomposition of $\ptraj(t)$.

 Similarly, using also the assumption $x(0)=\ptraj(0)$, 
\begin{equation}\label{eq:decompose xtilde into orthogonal directions 2} 
\begin{split}
x(t) &= \Big[\ptraj_w(0) + wt \Big]+ \left[\ptraj_Y(0) + \int_0^t \ff_Y(\tau)\,d\tau \right]\\
&= x_w(t) + x_Y(t), \
\end{split}
\end{equation}
with $\ff_Y(t)\in F_Y\big(x(t)  \big)$. 

Note that $x_Y$ is an unperturbed trajectory of the  FPCS system $F_Y$, on the $(n-1)$-dimensional subspace $Y$. Moreover, since $\tilde{\ff}_Y(t)\in F_Y\big( \ptraj(t) \big)$, $\ptraj_Y$ is a perturbed trajectory of the same system, associated with the perturbation $\Prt_Y(t)$. 
Since, $\Ltwo{\Prt_Y(\tau)}\le \Ltwo{\Prt(\tau)}\le \maxpert$, for all $\tau\ge 0$, it follows from the induction hypothesis that $\Ltwo{\ptraj_Y(t)-x_Y(t)}\le \divconst_\ldset\maxpert$ for $t<\tilde{T}$. Then, for $t< \tilde{T}$,
\begin{equation} \label{eq:the semi-final eq of local bound on x and xtilde!}
\begin{split}
\Ltwo{\ptraj(t)-x(t)} &\le \Ltwo{\ptraj_w(t) -x_w(t)} +  \Ltwo{\ptraj_Y(t) -x_Y(t)}\\
&\le \Ltwo{\Prt_w(t)}+\divconst_{\ldset} \maxpert\\
&\leq \maxpert + \big(\sigma-4 \big)\maxpert\\
&= \big(\sigma-3\big)\maxpert.
\end{split}
\end{equation}

The proof at this point would have been complete, except that 
in order to bound $\Ltwo{\ptraj(\tilde{T})-x(\tilde{T})}$, 
we need to account for the possibility that $U(t)$ has a jump at time $\tilde T$. We have
\ifOneCol
\begin{eqnarray}
\Ltwo{\ptraj(\tilde{T})-x(\tilde{T})} &\le& \limsup_{t\uparrow \tilde{T}} \Big( \Ltwo{\ptraj(\tilde{T})-\ptraj(t)} +\Ltwo{\ptraj(t)-x(t)} +\Ltwo{x(t)-x(\tilde{T})} \Big)\nonumber\\
&\le& \limsup_{t\uparrow \tilde{T}} \Ltwo{\ptraj(\tilde{T})-\ptraj(t)} \,+\, \limsup_{t\uparrow \tilde{T}}\Ltwo{\ptraj(t)-x(t)} \,+\, \limsup_{t\uparrow \tilde{T}}\Ltwo{x(t)-x(\tilde{T})}\nonumber\\
&\le& \limsup_{t\uparrow \tilde{T}}\left(\int_{t}^{\tilde{T}}\tilde{\ff}(\tau)\,d\tau + \Prt(\tilde{T}) - \Prt(t) \right) \,+\, \big(\sigma-3 \big)\maxpert \,+\, 0\nonumber\\
&\le& 2\maxpert + \big(\sigma-3\big)\maxpert \nonumber\\
&=&\sigma\maxpert -\maxpert, \label{eq:final eq of local bound on x and xtilde}
\end{eqnarray}
\else
\begin{eqnarray}
\Ltwo{\ptraj(\tilde{T})-x(\tilde{T})}& \le& \limsup_{t\uparrow \tilde{T}} \Big( \Ltwo{\ptraj(\tilde{T})-\ptraj(t)}\nonumber \\
&&+\,\Ltwo{\ptraj(t)-x(t)} +\Ltwo{x(t)-x(\tilde{T})} \Big)\nonumber\\
&\le& \limsup_{t\uparrow \tilde{T}} \Ltwo{\ptraj(\tilde{T})-\ptraj(t)} \, \nonumber\\
&&+\, \limsup_{t\uparrow \tilde{T}}\Ltwo{\ptraj(t)-x(t)} \, \nonumber\\
&&+\, \limsup_{t\uparrow \tilde{T}}\Ltwo{x(t)-x(\tilde{T})}\nonumber\\
&\le& \limsup_{t\uparrow \tilde{T}}\left(\int_{t}^{\tilde{T}}\tilde{\ff}(\tau)\,d\tau + \Prt(\tilde{T}) - \Prt(t) \right) \, \nonumber\\
&&+\, \big(\sigma-3 \big)\maxpert \,+\, 0\nonumber\\
&\le& 2\maxpert + \big(\sigma-3\big)\maxpert \nonumber\\
&=&\sigma\maxpert -\maxpert, \label{eq:final eq of local bound on x and xtilde}
\end{eqnarray}
\fi
where the third inequality is due to  (\ref{eq:the semi-final eq of local bound on x and xtilde!}) and the continuity of $x(t)$, and the last inequality is due to (\ref{eq:bounded pert at all times}) and the integrability of $\tilde{\ff}(\tau)$. 
Equation (\ref{eq:final eq of local bound on x and xtilde}) contradicts (\ref{eq:at Tr the distance iis equal to t}),
 and the lemma follows.
\end{proof}
\medskip

Suppose now that the perturbed trajectory stays far from the set of critical points throughout the time interval $[0,T]$. In light of Lemma \ref{lem:decreasing drift size}, we can divide $[0,T]$ into a finite number of subintervals during which the unperturbed system has a constant drift, use Lemma \ref{lem:locally close}  to obtain bounds on the distance of the perturbed and unperturbed trajectories during each subinterval, and then combine them to obtain a bound over the entire interval $[0,T]$. 

\begin{proposition} \label{prop:close to fluids at far}
Fix an FPCS system $F$.  Consider the constant  $\gamma$ in Lemma \ref{lem:distance from cp}, the constant $\sigma$ in Lemma \ref{lem:locally close}, and let $\saferad=m2^{m+1}\sigma$, where $m$ is the number of elements of the set $\actionset$ of drifts.  Let $x(\cdot)$ and $\ptraj(\cdot)$ be a pair of $\maxpert$-coupled trajectories at time $0$. If $d\big( \ptraj(t) , \CP \big)\ge \gamma\saferad\maxpert$ for all $t\in [0,T]$, then $\Ltwo{\ptraj(t)-x(t)}\le \saferad\maxpert$ for all $t\in[0,T]$.
\end{proposition}

\begin{proof}
As already mentioned, we will
 divide the  interval $[0,T]$  into at most $m{2^m}$ subintervals. We will then  use Lemma \ref{lem:locally close} to show that the distance between the two trajectories can only increase by an additive factor of  $\sigma\maxpert$ in each subinterval.  

We define a sequence of times $\tau_k$ by letting $\tau_1=0$ and
\begin{equation}\label{eq:def tau k}
\tau_{k+1} \triangleq \inf\left\{ t\in(\tau_k,T]\,\,\Big|\,\, \ldset_{k\sigma\maxpert}\big(x(t)\big)\not\subseteq  \ldset_{k\sigma\maxpert}\big(x(\tau_k)\big) \right\},
\end{equation}
for $k\ge 1$, 
with the convention that $\tau_{k+1}=T$ if the set on the right-hand side 
of \eqref{eq:def tau k} is empty. 
In words, $\tau_{k+1}$ is the time that the $(k\sigma\maxpert)$-neighbourhood of the unperturbed trajectory touches a new region, which does not intersect with the $(k\sigma\maxpert)$-neighbourhood of $x(\tau_k)$. Let $\Kmax$ be the maximum  $k$ such that $\tau_k<T$, so that
 $\tau_{\Kmax+1}=T$. For $k\le\Kmax$, we refer to the interval $[\tau_k, \tau_{k+1}]$ as phase $k$; see Fig. \ref{fig:phases far from critical}  for an illustration. 

\begin{figure*} 
\begin{center}
{
\ifOneCol
\includegraphics[width = 1\textwidth]{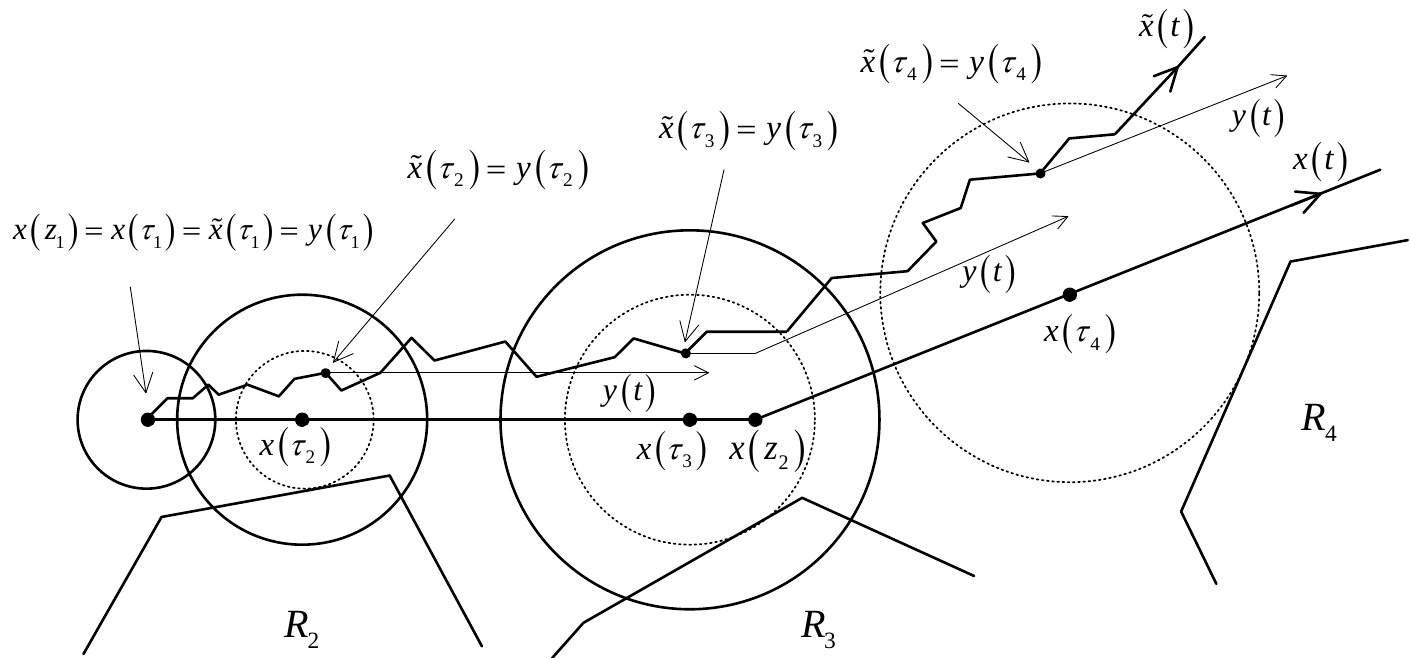}}
\vspace{-1cm}
\else
\includegraphics[width = .8\textwidth]{fig_phases_far_from_critical.pdf}}
	\vspace{-.2cm}
\fi
\end{center}
\caption{An illustration of the different phases and variables used in the proof of Proposition \ref{prop:close to fluids at far}
 (not all regions are shown). A new phase begins at time $\tau_{k+1}$  when the $(k\sigma\theta)$-neighbourhood of the unperturbed trajectory, $x(\cdot)$, touches a new region, $\region_{k+1}$. 
At the beginning of a phase, an auxiliary unperturbed trajectory $y(\cdot)$ is coupled with the perturbed trajectory $\ptraj(\cdot)$. The dotted circle is a translation of a solid circle,  centered at $x(t)$.  As soon as the dotted circle touches a new region boundary, a new solid circle, with larger radius is created.
 }
\label{fig:phases far from critical}
\end{figure*}

First, we show that the number of phases, $\Kmax$, is less than $m2^m$. 
According to Lemma \ref{lem:decreasing drift size}, the time interval $[0,T]$ can be partitioned into at most $2^{m}-1$ subintervals $[z_j,z_{j+1}),\,1\le j\le 2^m-1$, during each of which the unperturbed trajectory $x(t)$ is a line segment; that is, there exists a sequence of vectors $\dir_j,\,1\le j\le {2^m}-1$ such that
\begin{equation}\label{eq:linearly traverse a line segment}
 x(t)=x(z_j) + (t-z_j)\dir_j, \qquad \forall\ t\in [z_j, z_{j+1}].
\end{equation}
We argue that at most $m$ phase changes are possible during a  subinterval $[z_j,z_{j+1}]$, i.e., at most $m$ of the times $\tau_k$s lie in the interval $[z_j,z_{j+1}]$. Suppose that there are $l$ phase changes (for some $l\ge 0$), at times $\tau_{k_j+1}, \ldots, \tau_{k_j+l}\in(z_j,z_{j+1}]$. For each $k\in\big\{k_j+1,\ldots, k_j+l\big\}$, let $\mu_k\in\actionset$ be a drift that caused the phase change at time $\tau_{k}$, i.e.,
\begin{equation}\label{eq:the action that caused phase trans}
\mu_k\,\in\, \ldset_{(k-1)\sigma\maxpert}\big(x(\tau_k)\big)\, \backslash\, \ldset_{(k-1)\sigma\maxpert}\big(x(\tau_{k-1})\big).
\end{equation}
Equivalently,
\begin{equation}\label{eq:the action that caused phase trans regions}
d\big(x(\tau_{k-1}),\, \region_k\big) \,>\, d\big(x(\tau_k),\, \region_k\big) \,=\, (k-1)\sigma\maxpert,
\end{equation}
where $\region_k {\triangleq} \region_{\mu_k}$ is the effective region of $\mu_k$. We will now show that these regions  $\region_k$, for $k\in\big\{k_j+1,\ldots, k_j+l\big\}$, are distinct for different $k$. In order to draw a contradiction, suppose that there are  $k_1,k_2\in\big\{k_j+1,\ldots, k_j+l\big\}$ with $k_1<k_2$ such that $\mu_{k_1}=\mu_{k_2}$, or equivalently $\region\triangleq \region_{k_1}=\region_{k_2}$. 
Let  $f:[z_j,z_{j+1}]\to \R_+$ be the distance between $x(t)$ and the region $\region$:
\begin{equation}
f(t) \triangleq d\big( x(t)\,,\,  \region \big), \quad \forall\, t\in [z_j,z_{j+1}].
\end{equation}
The region $\region$ is a convex set. Therefore $f(\cdot)$  is the composition of a convex function (the distance from $\region$) and an affine function $x(t):[z_j,z_{j+1}]\to\R^n$ (see (\ref{eq:linearly traverse a line segment})). Hence, $f$ is also convex. Moreover, it follows from (\ref{eq:the action that caused phase trans regions}) and the assumption $k_1<k_2$ that
\begin{equation} \label{eq:fineq}
f(\tau_{k_2-1}) \,>\, f(\tau_{k_2}) \,=\, (k_2-1)\sigma\maxpert \,>\, (k_1-1)\sigma\maxpert \,=\, f(\tau_{k_1}).
\end{equation}
However, since $f$ is convex and $\tau_{k_1}\le\tau_{k_2-1}\le\tau_{k_2}$, we must have $f(\tau_{k_2-1})\le \max\big(f(\tau_{k_1})\, ,\, f(\tau_{k_2}) \big)$, which contradicts  \eqref{eq:fineq}. 
Hence,  
each distinct $k_i \in\big\{k_j+1,\ldots, k_j+l\big\}$ is associated with a disticnt region $R_{k_i}$.
On the other hand, since the number of different regions is at most $m$, there are at most $m$ phase changes during each of the at most $2^m-1$ line segments in the trajectory of $x(\cdot)$, and the total number of phases, $\Kmax$, is smaller than $m2^m$.

In our next step, we use induction on the phases to show that if $k\le \Kmax$, then
\begin{equation}\label{eq:induction on phases tau k}
\Ltwo{\ptraj(t)-x(t)} \le k\sigma\maxpert, \quad \forall\, t\in [\tau_k,\tau_{k+1}].
\end{equation}
For any $k\geq 1$, we consider the induction hypothesis
\begin{equation}\label{eq:induction hypo of induction on phases tau k}
\Ltwo{\ptraj(\tau_k)-x(\tau_k)} \le (k-1)\sigma\maxpert.
\end{equation}
Note that \eqref{eq:induction hypo of induction on phases tau k} is automatically true for $k=1$, because $\tau_1=0$ and $\ptraj(0)-x(0)$ has been assumed to be zero. This provides the basis of the induction.
Using the triangle inequality  and the inequalities $\gamma\ge1$ and $\eta=m2^{m+1}\sigma \geq 2\Kmax\sigma
\geq 2k\sigma$, we obtain 
\begin{equation}\label{eq:x is far far from CP}
\begin{split}
d\big(x(\tau_k),\CP\big) &\,\ge\, d\big(\ptraj(\tau_k),\CP\big) - 
\Ltwo{x(\tau_k) - \ptraj(\tau_k)}\\
&\,\ge\,\oneoverdel\saferad\maxpert  - (k-1)\sigma\maxpert\\
&\,\ge\, 2\gamma {k} \sigma\maxpert - (k-1)\sigma\maxpert\\
&\,\ge\, \gamma k \sigma\maxpert.
\end{split}
\end{equation}
Let  ${\ldset}=\ldset_{k\sigma\maxpert}\big(x(\tau_k)\big)$. It follows from (\ref{eq:x is far far from CP}) and Lemma \ref{lem:distance from cp}, with $r=k\sigma\theta$, that $\ldset$ is low-dimensional.
Furthermore, the definition of $\tau_{k+1}$ in (\ref{eq:def tau k}) implies that
\begin{equation} \label{eq:Uksigma x remains in subset of tauk}
\ldset_{k\sigma\maxpert}\big(x(t)\big) \subseteq {\ldset}, \qquad \forall\ t\in[\tau_k,\tau_{k+1}).
\end{equation}

Let $y(\cdot)$ be an unperturbed trajectory with initial condition $y(\tau_k)=\ptraj(\tau_k)$. Since the unperturbed dynamics are non-expansive, for any $t\ge \tau_k$, we have 
\ifOneCol
\begin{equation}\label{eq:y stays in k-1sigma neighbourhood of x}
\Ltwo{x(t)-y(t)}\le \Ltwo{x(\tau_k)-y(\tau_k)} = \Ltwo{x(\tau_k) - \ptraj(\tau_k)} \le (k-1)\sigma\maxpert.
\end{equation} 
\else
\begin{equation}\label{eq:y stays in k-1sigma neighbourhood of x}
\begin{split}
\Ltwo{x(t)-y(t)} &\,\le\, \Ltwo{x(\tau_k)-y(\tau_k)} \\
&\,=\, \Ltwo{x(\tau_k) - \ptraj(\tau_k)}\\
& \,\le\, (k-1)\sigma\maxpert.
\end{split}
\end{equation}
\fi
It is not hard to see that 
 (\ref{eq:y stays in k-1sigma neighbourhood of x}) and (\ref{eq:Uksigma x remains in subset of tauk}) imply that
\begin{equation}\label{eq:ball y inside ball x}
\ldset_{\sigma\maxpert}\big(y(t)\big) \subseteq \ldset,  \qquad \forall\ t\in[\tau_k,\tau_{k+1}),
\end{equation}
Hence, the conditions of Lemma \ref{lem:locally close} hold, with the initial time being $\tau_k$ instead of zero. Therefore, $\Ltwo{\ptraj(t)-y(t)}\le \sigma\maxpert$, for all $t\in[\tau_k,\tau_{k+1}]$. As a result, for $t\in [\tau_k,\tau_{k+1}]$,
\ifOneCol
\begin{equation}
\Ltwo{\ptraj(t)-x(t)}\le \Ltwo{\ptraj(t)-y(t)} + \Ltwo{y(t)-x(t)}\le \sigma\maxpert + (k-1)\sigma\maxpert = k\sigma \maxpert,
\end{equation}
\else
\begin{equation}
\begin{split}
\Ltwo{\ptraj(t)-x(t)}&\,\le\, \Ltwo{\ptraj(t)-y(t)} + \Ltwo{y(t)-x(t)}\\
 &\,\le\, \sigma\maxpert + (k-1)\sigma\maxpert \\
&\,=\, k\sigma \maxpert,
\end{split}
\end{equation}
\fi
where the second inequality is due to (\ref{eq:y stays in k-1sigma neighbourhood of x}). 
This establishes \eqref{eq:induction on phases tau k} and, in particular, that \eqref{eq:induction hypo of induction on phases tau k} holds with $k$ replaced by $k+1$ (the induction step). 
Finally,  the proposition follows from (\ref{eq:induction on phases tau k}) and the fact that $k\le\Kmax<m2^m$.
\end{proof}
\medskip

\subsection{Proof of the Bound when Close to a Critical Point}
\label{sub:c}
 In Proposition \ref{prop:close to fluids at far}, we presented a bound on the distance between  the trajectories 
when there are no nearby critical points. The next proposition deals with the other extreme, where the trajectories are in a basin of a critical point.

\begin{proposition}\label{prop:close in the basin}
Consider two constants $\maxpert,T>0$, a critical point $p\in\CP$ and a basin  $\ball_{\Rtwo}$ of radius $\Rtwo$ for $p$.
Let $x(\cdot)$ and $\ptraj(\cdot)$ be a pair of $\maxpert$-coupled trajectories at time $0$, with $\ptraj(t)\in \ball_{\Rtwo}$,  for all $t\in[0,T]$. Suppose that $0<\Rone<\Rtwo$, 
with $\Rtwo$ possibly infinite, 
and that for the ball $\ball_{\Rone}$ of radius $\Rone$ centered at $p$, 
\begin{equation}\label{eq:ring far from cp}
d\big(\ball_{\Rtwo}\backslash\ball_{\Rone}\,,\,\CP\big)\,  \ge \, \big(\gamma+1 \big) \saferad  \maxpert,
\end{equation}
where  $\saferad=m2^{m+1}\sigma$ and $\gamma$ are the constants defined in Proposition \ref{prop:close to fluids at far} and Lemma \ref{lem:distance from cp}, respectively. 
Then,  $\Ltwo{\ptraj(t)-x(t)}\le 4\Rone$, for all $t\in[0,T]$.
\end{proposition}

\begin{proof}
Since the critical point $p$ belongs to $\ball_{\Rtwo}$, \eqref{eq:ring far from cp}, implies that $p$ must belong to $\ball_{\Rone}$, and its distance from $\ball_{\Rtwo}\backslash\ball_{\Rone}$ is therefore at most $\Rone$. Hence, $d\big(\ball_{\Rtwo}\backslash\ball_{\Rone}\,,\,\CP\big)\leq \Rone$ and, in particular, $\Rone>\eta\theta$. 
Let
\begin{equation}
r_1 \triangleq \Rone-\saferad\maxpert ,\quad r_2\triangleq r_1+3\maxpert,
\end{equation}
and consider two balls $\ball_{r_1}$ and $\ball_{r_2}$ centered at $p$, with radii $r_1$ and $r_2$, respectively.  
Since $p$ is a critical point, it is in the intersection of at least two regions. Therefore,  the number  $m$ of elements of the set $\actionset$ of drifts is at least two, and 
\begin{equation} \label{eq:eta16}
\saferad =m2^{m+1}\sigma\ge m2^{m+1} \ge 16.
\end{equation}
As a result, 
 $r_2\le \Rone $ and $\ball_{r_1}\subset \ball_{r_2}  {\subset}\, \ball_{\Rone}\subset\ball_{\Rtwo}$;
see Fig.~\ref{fig:proof of the dragging lemma}.

\begin{figure*} 
\begin{center}
{
\ifOneCol
\includegraphics[width = 1\textwidth]{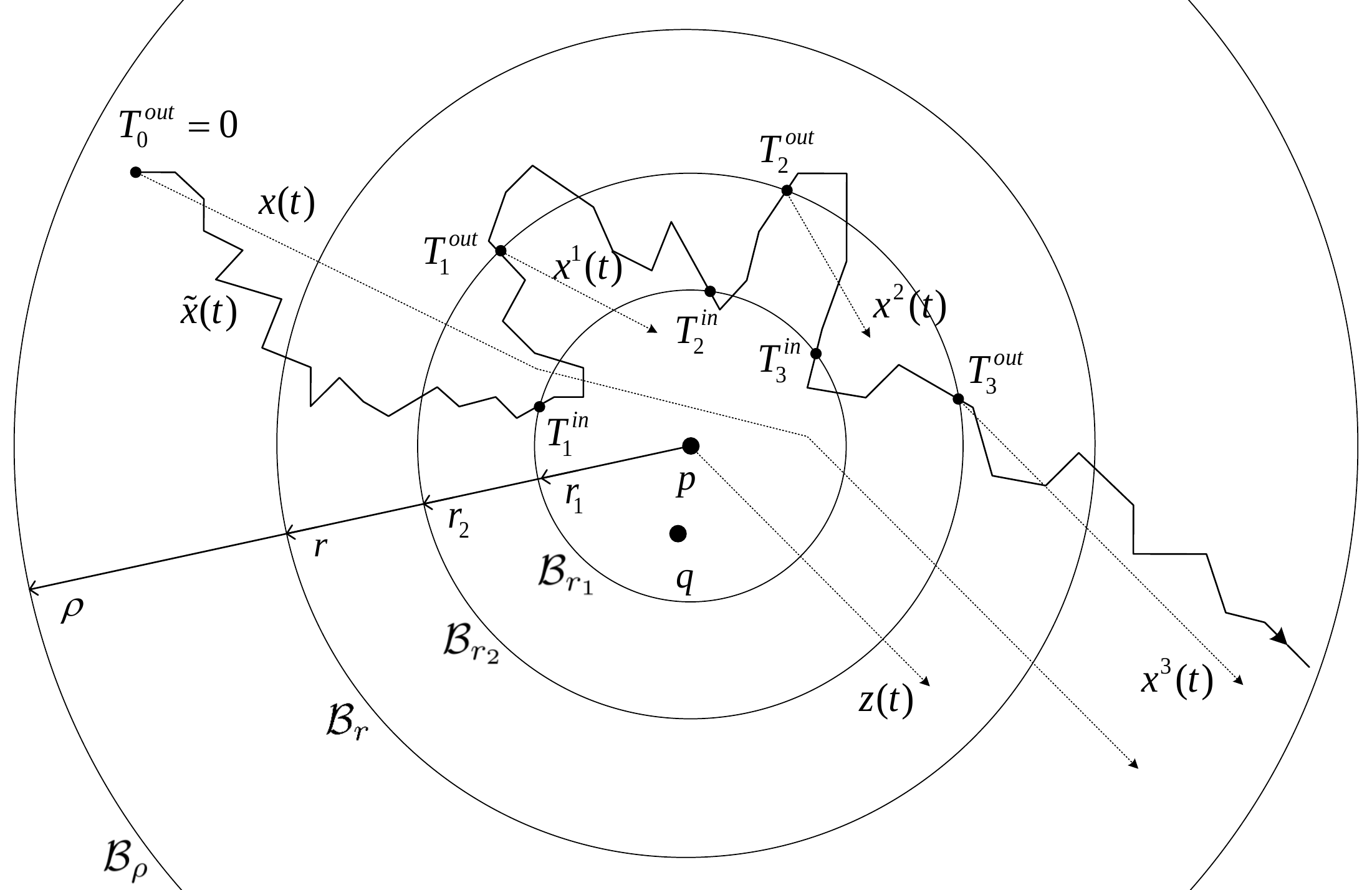}}
\vspace{-1cm}
\else
\includegraphics[width = .65\textwidth]{fig_phases_of_close_to_basin.pdf}}
\vspace{-.2cm}
\fi
\end{center}
\caption{An illustration of the balls, the different rounds, and the variables used in the proof of Proposition \ref{prop:close in the basin}. Here, $p$ and $q$ are two critical points, and $T_i^{in}$ and $T_i^{out}$ are defined in (\ref{eq:def Tin}) and (\ref{eq:def Tout}), respectively. There are four balls $\ball_{r_1}$, $\ball_{r_2}$, $\ball_{\Rone}$, and $\ball_{\Rtwo}$ all centered at $p$, with radii $r_1$, $r_2$, $\Rone$, and $\Rtwo$, respectively. Also, $x(\cdot)$ is an unperturbed trajectory that is coupled with the perturbed trajectory $\ptraj(\cdot)$ at time $0$, $z(\cdot)$ is an unperturbed trajectory that starts at $p$ at time $T_1^{in}$, and each $x^i(\cdot)$ is an unperturbed trajectory that is coupled with $\ptraj(\cdot)$ at time $T_i^{out}$.}
\label{fig:proof of the dragging lemma}
\end{figure*}

Consider the region $\ball_{\Rtwo}\backslash \ball_{r_1}$
 between $\ball_{r_1}$ and $\ball_{\Rtwo}$. 
It follows from (\ref{eq:ring far from cp}) that
\ifOneCol
\begin{equation}\label{eq:r ring far from cp}
d\big(\ball_{\Rtwo}\backslash \ball_{r_1}\,,\,\CP\big)\,  \ge \, d\big(\ball_{\Rtwo}\backslash \ball_{\Rone}\,,\,\CP\big) - \big(\Rone-r_1\big)\, \ge\, \big(\gamma  + 1 \big) \saferad\maxpert- \saferad\maxpert \,=\, \gamma  \saferad \maxpert.
\end{equation}
\else
\begin{equation}\label{eq:r ring far from cp}
\begin{split}
d\big(\ball_{\Rtwo}\backslash \ball_{r_1}\,,\,\CP\big)&\,  \ge \, d\big(\ball_{\Rtwo}\backslash \ball_{\Rone}\,,\,\CP\big) - \big(\Rone-r_1\big)\\
&\, \ge\, \big(\gamma  + 1 \big) \saferad\maxpert- \saferad\maxpert \\
&\,=\, \gamma  \saferad \maxpert.
\end{split}
\end{equation}
\fi

The high level idea is that for any $t\in[0,T]$, the perturbed solution is either in $\ball_{r_2}$ or in $\ball_{\Rtwo}\backslash \ball_{r_1}$ (or possibly in both). When $\ptraj(t)\in \ball_{r_2}$ we will show that the unperturbed trajectory is also close to the critical point $p$. But the more interesting case is when $\ptraj(t)\in \ball_{\Rtwo}\backslash \ball_{r_1}$. The idea here is to look at the perturbed solution, and 
at certain times that 
 it hits the boundary of $\ball_{r_2}$,  
 consider an auxiliary unperturbed trajectory that is coupled with $\ptraj(t)$ at that time. Using Proposition \ref{prop:close to fluids at far}, we can then show that these coupled trajectories stay close to each other, as long as the perturbed trajectory stays in $\ball_{\Rtwo}\backslash \ball_{r_1}$. As a result, and using also the fact that the dynamics are non-expansive, the auxiliary trajectory  $\ptraj(t)$ will remain close to $x(t)$, and the distance $\Ltwo{\ptraj(t)-x(t)}$ stays bounded. 
The various parameters and trajectories are illustrated in Fig. \ref{fig:proof of the dragging lemma}.

Let $T_0^{out}=0$, and for any $i\ge1$ let 
\begin{equation}\label{eq:def Tin}
T_i^{in} \triangleq \inf\Big\{t\in ( T_{i-1}^{out},T]\,\,\Big|\,\, \ptraj(t)\in \ball_{r_1}\Big\},
\end{equation}
\begin{equation}\label{eq:def Tout}
T_i^{out} \triangleq \inf\Big\{t\in ( T_i^{in}, T]\,\,\Big|\,\, \ptraj(t)\notin \ball_{r_2}\Big\}.
\end{equation} 
If either set is empty, we let the left hand side be equal to $T$. We consider a number of rounds. Round $i$ starts at time $T_i^{out}$ and ends at time $T_{i+1}^{in}$.  
 Note that the union of these rounds does not necessarily cover $[0,T]$. Also, note that since there is a gap of size $3\maxpert>  2\maxpert$  
between the boundaries of $\ball_{r_1}$ and $\ball_{r_2}$, it takes some lower-bounded positive time for the perturbed trajectory to travel from one boundary to the other, and hence the length of each round is lower bounded by a positive constant. 
So, the number of rounds during $[0,T]$ is finite. 

To each round $i$ we associate an unperturbed trajectory, denoted by $x^i(t),\,t\in[T_i^{out},T_{i+1}^{in}]$, with initial point $x^i(T_i^{out})=\ptraj(T_i^{out})$, i.e., $x^i(\cdot)$ is coupled with the perturbed trajectory at time $T_i^{out}$.
For any $t\in[T_i^{out},T_{i+1}^{in}]$, since $\ptraj(t)\in\ball_{\Rtwo}\backslash\ball_{r_1}$, (\ref{eq:r ring far from cp}) asserts that  $d\big(\ptraj(t),\,\CP\big)\ge \gamma\saferad\maxpert$. Therefore,  it follows from Proposition \ref{prop:close to fluids at far} that  for any $t\in[T_i^{out},T_{i+1}^{in}]$, 
\begin{equation}\label{eq:ptraj near xi in the phases}
\Ltwo{\ptraj(t)-\x^i(t)}\le \saferad\maxpert.
\end{equation}

Note that $x^0(t)=x(t)$, for all $t\ge0$. Thus, if $T_1^{in}=T$, 
the inequality (\ref{eq:ptraj near xi in the phases}) together with the fact $\eta\theta <\Rone<4\Rone$ imply that $\Ltwo{\ptraj(t)-x(t)}< 4\Rone$, for all $t\in[0,T]$, as desired. 
 So, in the following we assume that $T_1^{in}<T$. 
Note that the right-continuity of $\ptraj(\cdot)$ implies that $\Ltwo{\ptraj(T_1^{in})-p}\le r_1$.  
 Let $z(\cdot)$ be an unperturbed trajectory that starts at the critical point $p$ at time $T_1^{in}$, i.e., $z\big(T_1^{in}\big)=p$. 
It follows from (\ref{eq:ptraj near xi in the phases}) and the non-expansive property of the dynamics that for any $t\ge T_1^{in}$,
\begin{equation}\label{eq:bound of the original fluid and the centered fluid}
\begin{split}
\Ltwo{\x(t)-z(t)} &\,\le\, \Ltwo{\x(T_1^{in})-z(T_1^{in})}\\
&\,=\, \Ltwo{\x(T_1^{in})-p}\\
&\,\le\, \Ltwo{p-\ptraj(T_1^{in})} + \Ltwo{\ptraj(T_1^{in})-\x(T_1^{in})}\\
&\,\le\, r_1 + \saferad\maxpert\\
&\,=\,\Rone,
\end{split}
\end{equation}
where in the last inequality we used  \eqref{eq:ptraj near xi in the phases} with $i=0$ and  $t=T_1^{in}$. We now proceed  
to derive a bound on $\Ltwo{\ptraj(t)-x(t)}$, for $t\ge T_1^{in}$, by developing a bound on $\Ltwo{\ptraj(t)-z(t)}$.
Let $\dir=\dir(p)$ be the actual drift at $p$. We consider two cases.

{\bf Case 1.} ($p$ is an equilibrium point, i.e.,  $\dir=0$){\bf .}\,  
In this case,  $z(t)=p$, for all $t\geq T_{1}^{in}$. Comparing with the unperturbed trajectory $\x^i(\cdot)$ and using
the non-expansive property and the definition of $T_i^{out}$, it follows for any round $i\ge1$ and any $ t\in[T_i^{out} ,T_{i+1}^{in}]$, we have
\ifOneCol
\begin{equation}\label{eq:local fluid bound when near zero case 1}
\Ltwo{\x^i(t)-z(t)} \,\le\, \Ltwo{\x^i(T_i^{out})-z(T_i^{out})} \,=\, \Ltwo{\ptraj(T_i^{out})-p} \,=\, r_2\,\le \, 3r_2.
\end{equation}
\else
\begin{equation}\label{eq:local fluid bound when near zero case 1}
\begin{split}
\Ltwo{\x^i(t)-z(t)} &\,\le\, \Ltwo{\x^i(T_i^{out})-z(T_i^{out})}\\
& \,=\, \Ltwo{\ptraj(T_i^{out})-p}\\
& \,=\, r_2\,\le \, 3r_2.
\end{split}
\end{equation}
\fi
Combining this with (\ref{eq:ptraj near xi in the phases}) and (\ref{eq:bound of the original fluid and the centered fluid}), we get the following bound for all  $i\ge 1$ and for all $t\in[T_i^{out} ,T_{i+1}^{in}]$,
\begin{equation}\label{eq:case 1 temp1}
\begin{split}
\Ltwo{\ptraj(t)-\x(t)} &\,\le\, \Ltwo{\ptraj(t)-x^i(t)} + \Ltwo{x^i(t)-z(t)} 
\ifOneCol 
+ \Ltwo{z(t)-\x(t)}\\
\else
\\&\quad\,\, + \Ltwo{z(t)-\x(t)}\\
\fi
&\,\le\, \saferad\maxpert + 3r_2 + \Rone\\
&\, =\, \saferad\maxpert \,+\, 3\big( \Rone - \saferad\maxpert +3\maxpert  \big) \,+\, \Rone\\
&\,=\, 4\Rone +2\big(4.5 -\saferad\big)\maxpert\\
&\,<\, 4\Rone,
\end{split}
\end{equation}
where the second inequality is due to (\ref{eq:ptraj near xi in the phases}), (\ref{eq:local fluid bound when near zero case 1}), and (\ref{eq:bound of the original fluid and the centered fluid}), and the last inequality is due to \eqref{eq:eta16}. 

Furthermore, for any $t\in[T_i^{in},T_i^{out})$, our definitions imply that $\Ltwo{\ptraj(t)-z(t)} =\Ltwo{\ptraj(t)-p}<r_2$. Hence, for such $t$,
\begin{equation*}
\Ltwo{\ptraj(t)-\x(t)} \le \Ltwo{\ptraj(t)-z(t)} + \Ltwo{z(t)-\x(t)} \le r_2+\Rone<4\Rone,
\end{equation*}
where the second inequality is due to (\ref{eq:bound of the original fluid and the centered fluid}). Thus, the proposition holds in this case.

{\bf Case 2.} ($p$ is not an equilibrium point, i.e.,  $\dir\ne 0$){\bf .}\,
 The dynamics in this case are illustrated in Fig. \ref{fig:proof of the dragging lemma}.  
Here, we need to find an alternative  
derivation of (\ref{eq:local fluid bound when near zero case 1}), and also derive a new bound for $\Ltwo{\ptraj(t)-z(t)}$ when $t\in[T_i^{in},T_i^{out})$. 

Let $\ff(\cdot)$ be a perturbed drift associated with the perturbed trajectory $\ptraj(\cdot)$. Since $\ball_{\Rtwo}$ is a basin and $\ff(t)\in F\big(\ptraj(t)\big)$, it follows from the definition of basins that for any $t\in [0,T]$,
\begin{equation}\label{eq:derivative of xit ptraj}
\dir^T\ff(t)  \ge \Ltwo{\dir}^2.
\end{equation}
Hence, for any $t\in [T_1^{in},T]$,
\begin{equation}\label{eq:lower bound on distance of ptraj and p to be used for tesc}
\begin{split}
\Ltwo{\ptraj(t)-p} &\ge \frac{1}{\Ltwo{\dir}}\dir^T\big(\ptraj(t)-p\big) \\
& \,=\,  \frac1{\Ltwo{\dir}}\Bigg( \dir^T\Big(\ptraj\big(T_1^{in}\big)-p\Big)\, +\,  \int_{T_1^{in}}^t \dir^T\ff(\tau)\,d\tau 
\ifOneCol
\,+\, \dir^T\big(\Prt(t)-\Prt(T_1^{in})  \big) \Bigg)\\
\else
\\&\qquad\qquad  \,+\, \dir^T\big(\Prt(t)-\Prt(T_1^{in})  \big) \Bigg)\\
\fi
&\,\ge\, -\Ltwo{\ptraj\big(T_1^{in}\big)-p} + \frac1{\Ltwo{\dir}} \int_{T_1^{in}}^t \Ltwo{\dir}^2 \,d\tau \,-\, 2\maxpert \\
&\,\ge\,  -r_1+ (t-T_1^{in})\Ltwo{\dir}  - 2\maxpert .
\end{split}
\end{equation}
The first inequality above is the Cauchy-Schwarz inequality; the first equality  follows from the definition of perturbed trajectories (cf.~Definition \ref{def:integral pert traj}); the next inequality uses the Cauchy-Schwarz inequality for the first term,  (\ref{eq:derivative of xit ptraj}) for the second, and the bounds on $U(\cdot)$ for the third; the last inequality uses the defining property
$\Ltwo{\ptraj\big(T_1^{in}\big)-p}=r_1$ of $T_1^{in}$.
 
We define an escape time $T^{esc} = T_1^{in} + \big(r_1+r_2+3\maxpert\big)/{\lVert{\dir}\rVert}$. The following claim suggests that if $\ptraj(t)$ ever escapes $\ball_{r_2}$, it happens before time $T^{esc}$.
\begin{claim}\label{claim:tesc}
If $T_i^{in}<T$ for some $i\ge1$, then $T_i^{out}\le T^{esc}$.
\end{claim} 
\begin{proof}[Proof of Claim]
If $T^{esc}\ge T$, then $T_{i}^{out}\le T  \le T^{esc}$. Suppose now that $T^{esc}< T$. It follows from  (\ref{eq:lower bound on distance of ptraj and p to be used for tesc}) and the  definition of $T^{esc}$ that $\Ltwo{\ptraj\big(T^{esc}\big)-p}\ge r_2+\theta> r_2$.
Hence $\ptraj(t)$ is outside of $\ball_{r_2}$ at time $T^{esc}$, and by definition, $T_i^{out}\le T^{esc}$.
\end{proof}

Since $z(T_1^{in})=p$, we have $\Ltwo{\dot{z}(0)}=\dir$.
According to Lemma \ref{lem:decreasing drift size}, $\Ltwo{\dot{z}(t)}$ is a non-increasing function of time, and therefore 
 $\Ltwo{\dot{z}(t)}\le\Ltwo{\dir}$, for $t\geq T_1^{in}$.
Hence, 
if $t\in\big[T_{1}^{in},T^{esc}\big]$, then 
\ifOneCol
\begin{equation}\label{eq:dist z bound for t small}
\Ltwo{z(t)-p} \,\le \, \big(t - T_1^{in}\big) \Ltwo{\dir} \,\le\, \big(T^{esc} - T_1^{in}\big) \Ltwo{\dir}\,=\, r_2+r_1+3\maxpert\, =\, 2r_2.
\end{equation}
\else
\begin{equation}\label{eq:dist z bound for t small}
\begin{split}
\Ltwo{z(t)-p} &\,\le \, \big(t - T_1^{in}\big) \Ltwo{\dir}\\
&\,\le\, \big(T^{esc} - T_1^{in}\big) \Ltwo{\dir}\\
&\,=\, r_2+r_1+3\maxpert\\
&\, =\, 2r_2.
\end{split}
\end{equation}
\fi

We now proceed by considering two cases: $t\in\big[T_{i}^{out},T_{i+1}^{in}\big)$ and $t\in\big[T_{i}^{in},T_{i}^{out}\big)$.
We first suppose that $T_i^{out}<T$, and consider a $t\in \big[T_{i}^{out},T_{i+1}^{in}\big)$. Since $T_i^{in}\le T_i^{out} < T$,  Claim~\ref{claim:tesc} implies that $T_i^{out}\le T^{esc}$. It then follows from \eqref{eq:dist z bound for t small} that $\Ltwo{z\big(T_i^{out})-p}\le 2r_2$. Then,
\begin{equation}\label{eq:dist xi and z for i le i*}
\begin{split}
\Ltwo{x^{i}\big(t\big)-z\big(t\big)} &\,\le\, \Ltwo{x^{i}\big(T_{i}^{out}\big)-z\big(T_{i}^{out}\big)}\\
 &\,\le\, \Ltwo{x^{i}\big(T_{i}^{out}\big)-p}  \,+\,  \Ltwo{p-z\big(T_{i}^{out}\big)} \\
&\,\le\, r_2\,+\, 2r_2\\
& \,=\, 3r_2.
\end{split}
\end{equation}
where the first inequality is due to the non-expansive property, and the last inequality is due to the definition of $T_{i}^{out}$. 
Thus, the bound (\ref{eq:local fluid bound when near zero case 1})  and the subsequent derivation of (\ref{eq:case 1 temp1}) remain valid for this case as well,
so that for every round $i$, 
\begin{equation}\label{eq:case 2 temp1}
\begin{split}
\Ltwo{\ptraj(t)-\x(t)} < 4\Rone,\qquad \forall\ t\in\big[T_{i}^{out},T_{i+1}^{in}\big).
\end{split}
\end{equation}

We now discuss the case where $t$ does not belong to a round, i.e., $t\in\big[T_i^{in}, T_{i}^{out}\big)$, for some $i$ such that $T_i^{in}<T$. It follows from Claim~\ref{claim:tesc} that $t<T_i^{out}\le T^{esc}$. Therefore, (\ref{eq:dist z bound for t small}) implies that if 
$t\in\big[T_i^{in}, T_{i}^{out}\big)$, then 
$t\Ltwo{p-z(t)}\le 2r_2$.
Therefore,
\begin{equation*}
\begin{split}
\Ltwo{\ptraj(t)-\x(t)} &\le \Ltwo{\ptraj(t)-p} + \Ltwo{p-z(t)} + \Ltwo{z(t)-\x(t)}\\
&\le r_2\,+\, 2r_2 \,+\, \Rone\\
&< 4\Rone,
\end{split}
\end{equation*}
where the second inequality is due to $t\in\big[T_{i}^{in},T_i^{out}\big)$
and (\ref{eq:bound of the original fluid and the centered fluid}). 
Together with (\ref{eq:case 2 temp1}), this completes 
the argument for Case 2, and the proof of the proposition. 
\end{proof}

\subsection{Completing the Proof of the Theorem} \label{subsec:Completing the Proof of the Theorem}
\label{sub:d}
We now use the machinery developed in this section and combine the results for the various cases to complete the proof of  Theorem \ref{th:main cont}.

If there are no critical points, then the perturbed trajectory  never gets close to a critical point, and the theorem follows from Proposition  \ref{prop:close to fluids at far}. In the following, we assume that the set of critical points is non-empty.
 Let $M$ be the number of critical points, 
let $\basinmin$ be the CNC, 
and let $D^{\CP}$ be the diameter of the set of critical points:
\begin{equation}
D^{\CP}\triangleq \max_{p,q\in\CP} \Ltwo{p-q}.
\end{equation}
According to Lemma \ref{lem-cp: finite number}, there are finitely many critical points, so that $D^{\CP}$ is well-defined and finite.
We define a threshold  parameter $\ctresh$ 
as follows:
\begin{equation}\label{eq:ctresh}
\ctresh\triangleq \frac{\basinmin}{40\big(M+2\big)\big(\gamma+1\big)\saferad},
\end{equation}
where $\gamma$ and $\saferad$ are the constants defined in Lemma \ref{lem:distance from cp} and Proposition \ref{prop:close to fluids at far}, respectively.
In what follows, we use Proposition \ref{prop:close in the basin} to prove that the following constant satisfies Theorem~\ref{th:main cont}, 
\begin{equation}\label{eq:def divconst in the proof of cont th}
\divconst \,=\, \begin{cases} {4D^\CP}/{\ctresh}\,+\,5\big(M+2\big)\big(\gamma+1\big)\saferad, & \textrm{if } \ctresh\ne 0, \\
4\big(\gamma+1\big)\saferad+1,  & \textrm{if } \ctresh = 0.\end{cases} 
\end{equation}
As  an example, for the system illustrated in Fig. \ref{fig:application finite part is subdiff}, it can be checked that the above constants are as follows: $D^\CP=\ctresh=0$, $M=1$, $\gamma=1$, $\sigma=5$,  $\saferad=240$, and  $\divconst=1921$. 

We consider two cases, depending on whether the perturbation bound $\maxpert$ is larger or smaller than the threshold $\ctresh$.

{\bf Case 1} ($\maxpert\ge \ctresh$){\bf.}\, 
According to Lemma \ref{lem-cp: basin infinite rad}, there exists a critical point $p^*$, for which the entire set $\R^n$ is a basin.
We let
$\Rone=D^{\CP}+\big(\gamma+1\big)\saferad\maxpert$ and $\Rtwo=\infty$. 
This choice of $p^*$, $\Rone$, and $\Rtwo$ observes  the conditions of Proposition \ref{prop:close in the basin}. 
Note that if $\ctresh=0$, then $\basinmin=0$, in which case there is at most one critical point. Then, $\ctresh=0$ implies $D^\CP=0$.  Therefore, $4\Rone < \divconst\maxpert$, for all values of $\ctresh$.
It then follows from Proposition \ref{prop:close in the basin} that 
\begin{equation}
\Ltwo{\ptraj(t)-
{x}(t)} \le 4\Rone < \divconst\maxpert, \qquad \forall t\ge0,
\end{equation}
which establishes the desired result.

\medskip

{\bf Case 2} ($\maxpert< \ctresh$){\bf.}\, 
Once more, we rely on Proposition \ref{prop:close in the basin}, but in a local manner. 
We consider a ``small'' basin of size $\basinmin/2$ for each critical point, and define a number of phases $\big[T_i^{in},T_i^{out}  \big]$ so that 
throughout any particular  phase, the perturbed trajectory lies in one of these basins.  We then use Proposition \ref{prop:close in the basin} to bound the distance between the two trajectories in each phase, and use Proposition \ref{prop:close to fluids at far} to bound their distance while outside the basins.  In the end, we use Lemma \ref{lem-cp: no revisit} to show that each  basin is visited at most once, in a certain  sense, and finally put everything together to prove the desired bound on the distance of the two trajectories. Figure \ref{fig:proof of th continuous for case two jumping between critical points} shows an illustration of the different trajectories and  variables that we use in the argument that follows.

\begin{figure*} 
\begin{center}
{
\ifOneCol
\includegraphics[width = 1\textwidth]{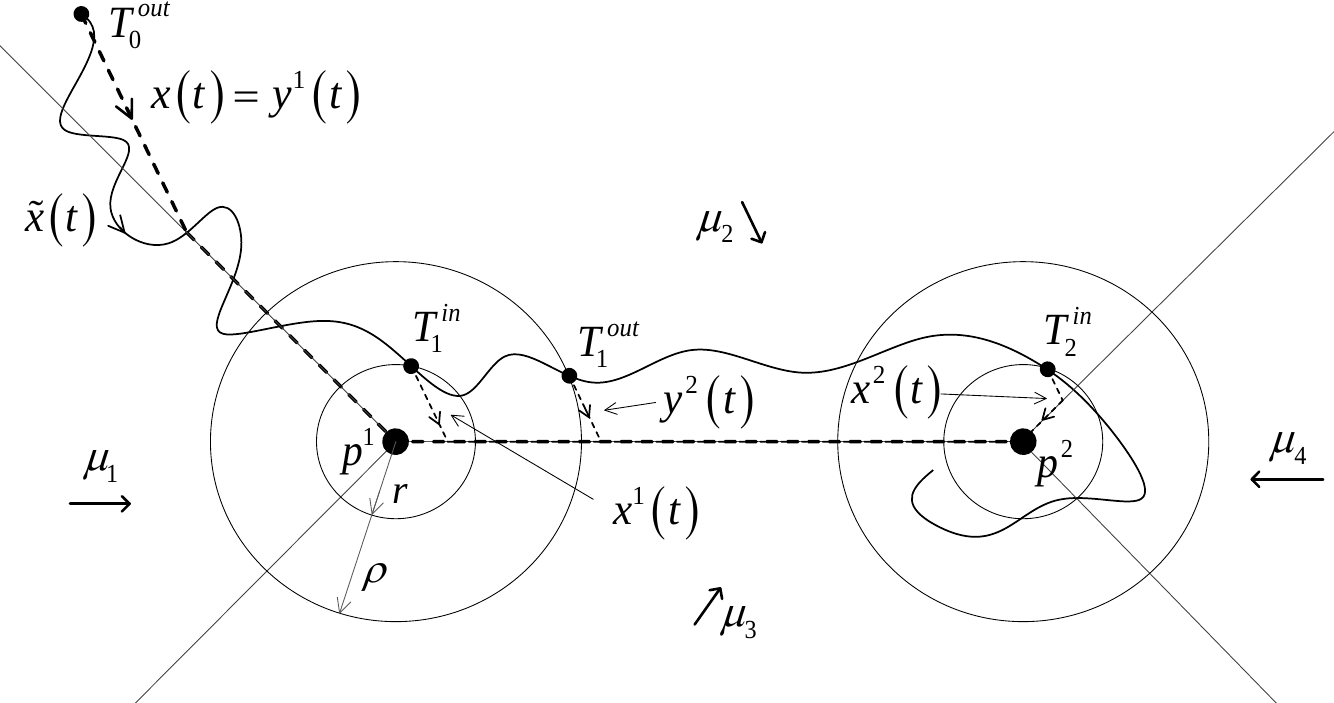}}
\vspace{-1.5cm}
\else
\includegraphics[width = .7\textwidth]{fig_proof_th_cont_case2.pdf}}
\vspace{-.3cm}
\fi
\end{center}
\caption{Illustration of the different trajectories and variables in Case 2 of the proof of Theorem \ref{th:main cont}. The figure shows a two-dimensional FPCS system with four regions and two critical points, $p^1$ and $p^2$. There are two balls of radii $\Rone$ and $\Rtwo$, defined in (\ref{eq:def of R1 and R2 for case 2 proof th cont}), centered at each critical point. The solid curved line $\ptraj(\cdot)$ is a perturbed trajectory, coupled with an unperturbed trajectory $x(\cdot)$ at time $0$. Times $T_i^{in}$ and $T_i^{out}$, defined in (\ref{eq:def Tin of pi})  and (\ref{eq:def Tout of pi}), are the first times that the perturbed trajectory hits the 
ball of radius $\Rone$ or leaves the larger ball of radius $\Rtwo$ around $p^i$, respectively. 
In this example, $T_2^{out}=\infty$,  
because the perturbed trajectory never leaves the $\Rtwo$-neighbourhood of $p^2$. For each $i$, $x^{i}(\cdot)$ and $y^i(\cdot)$ are unperturbed trajectories coupled with $\ptraj(\cdot)$ at times $T_i^{in}$ and $T_{i-1}^{out}$, respectively. These unperturbed trajectories are shown by dashed lines.
}
\label{fig:proof of th continuous for case two jumping between critical points}
\end{figure*}

Let 
\begin{equation}\label{eq:def of R1 and R2 for case 2 proof th cont}
\Rone=\big(\gamma+1\big)\saferad\maxpert,\qquad \Rtwo={\basinmin}/{2}.
\end{equation}
It follows from  (\ref{eq:ctresh}) and the assumption $\maxpert< \ctresh$ that $\Rone<\Rtwo$. Moreover, based on Lemma \ref{lem-cp: basin cnc}, $\Rtwo$ is a basin radius for each one of the critical points. For any critical point $p\in\CP$, let  $\ball_{\Rone}(p)$  and $\ball_{\Rtwo}(p)$ be the balls of radii $\Rone$ and $\Rtwo$, respectively, centered at $p$.  We define two sequences of times $T_i^{in}$ and $T_i^{out}$ as follows.
Let $T_0^{out}=0$, and for any $i\ge1$, let
\begin{equation}\label{eq:def Tin of pi}
T_i^{in} \triangleq \inf\Big\{t>T_{i-1}^{out}\,\,\Big|\,\, \exists\, 
p\in \CP:\,\, \ptraj(t)\in \ball_{\Rone}({p})\Big\}.
\end{equation}
We denote by $p^i$ the
critical point $p$  in the right-hand side of \eqref{eq:def Tin of pi}, so that 
$\ptraj(T_i^{in})\in \ball_{\Rone}(p^i)$. Note that  the different balls $\ball_{\Rone}(p)$ do not intersect and therefore $p^i$ is uniquely defined; we refer to it as the \emph{effective critical point} at phase $i$. We then define
\begin{equation}\label{eq:def Tout of pi}
T_i^{out} \triangleq \inf\Big\{t>T_i^{in}\,\,\Big|\,\, \ptraj(t)\notin \ball_{\Rtwo}(p^i) \Big\}.
\end{equation}
In the above, we let $T_i^{in}$ or $T_i^{out}$ be infinite in case the set on the right-hand side of \eqref{eq:def Tin of pi} or \eqref{eq:def Tout of pi} is empty.


Fix an $i\ge1$. We first derive a bound on $\Ltwo{\ptraj(t)-x(t)}$ for $t\in \big[T_{i-1}^{out},T_{i}^{in}\big]$.
Let $y^i(\cdot)$ be an unperturbed trajectory with $y^i\big(T_{i-1}^{out}\big) = \ptraj\big(T_{i-1}^{out}\big)$. 
By definition, for any $t\in \big[T_{i-1}^{out},T_{i}^{in}\big]$, $d\big(\ptraj(t),\CP\big)\ge \Rone= \big(\gamma+1 \big) \saferad  \maxpert $. Therefore, it follows from Proposition \ref{prop:close to fluids at far} that
\begin{equation}
\Ltwo{\ptraj(t) - y^i(t)} \le \saferad  \maxpert < \Rone, \quad \forall\, t\in\big[T_{i-1}^{out},T_{i}^{in}\big].
\end{equation}
Hence, from the non-expansive property of the unperturbed dynamics, for any $t\in \big[T_{i-1}^{out},T_{i}^{in}\big]$, we obtain
\begin{equation}\label{eq:dist increase throughout an out to in phase}
\begin{split}
\Ltwo{\ptraj(t) - x(t)} &\,\le\, \Ltwo{x(t) - y^i(t)} + \Ltwo{\ptraj(t) - y^i(t)} \\
&\,\le\, \Ltwo{x\big(T_{{i-1}}^{out}\big) - y^i\big(T_{{i-1}}^{out}\big)} +\Rone \\
&\,=\, \Ltwo{ x \big(T_{{i-1}}^{out}\big)- {\ptraj}\big(T_{{i-1}}^{out}\big) } +\Rone .
\end{split}
\end{equation}

On the other hand, if $t\in \big[T_i^{in},T_{i}^{out}\big]$, for some $i\geq 1$, we have for any critical point $p\in\CP$ other than $p^i$,
\ifOneCol
\begin{equation}
d\big(p,\,\ball_{\Rtwo}(p^i)  \big)
\,  \ge\,\Ltwo{ p - p^i  }- \Rtwo \,\ge\, \basinmin - \Rtwo \,=\, \frac{\basinmin}{2} \,\ge\, \big(\gamma+1\big)\saferad\ctresh \,>\, \Rone,
\end{equation}
\else
\begin{equation}
\begin{split}
d\big(p,\,\ball_{\Rtwo}(p^i)  \big)
&\,  \ge\,\Ltwo{ p - p^i  }- \Rtwo \\
&\,\ge\, \basinmin - \Rtwo\\ 
&\,=\, \frac{\basinmin}{2} \\
&\,\ge\, \big(\gamma+1\big)\saferad\ctresh \\
&\,>\, \Rone,
\end{split}
\end{equation}
\fi
where the second and third inequalities follow from the definitions of $\basinmin$ and $\ctresh$ in (\ref{eq:def basinmin}) and (\ref{eq:ctresh}), respectively.  
Using also the fact $d\big(p^i,\,\, \ball_{\Rtwo}(p^i)\backslash\ball_{\Rone}(p^i)\big)\geq \Rone$, we obtain
\begin{equation}
d\big(\ball_{\Rtwo}(p^i)\backslash\ball_{\Rone}(p^i)\,,\,\CP\big)\,  \ge\,   \Rone \,= \, \big(\gamma+1 \big) \saferad  \maxpert.
\end{equation}
Moreover, for any $t\in \big[T_i^{in},T_{i}^{out}\big]$, $\ptraj(t)\in \ball_{\Rtwo}(p^i)$.
Hence, the conditions of Proposition \ref{prop:close in the basin} are observed. 
Let $x^i(t)$ be an unperturbed trajectory with $x^i\big(T_i^{in}\big) = \ptraj\big(T_i^{in}\big)$.  Proposition~\ref{prop:close in the basin} implies that, for any $t\in \big[T_i^{in},T_{i}^{out}\big]$, 
$\Ltwo{\ptraj(t) - x^i(t)} \le 4\Rone$.
Hence, from the non-expansiveness of the dynamics, we have for any $t\in \big[T_i^{in},T_{i}^{out}\big]$, 
\begin{equation}\label{eq:dist increase throughout an in to out phase}
\begin{split}
\Ltwo{\ptraj(t) - x(t)} &\,\le\, \Ltwo{x(t) - x^i(t)} + \Ltwo{\ptraj(t) - x^i(t)} \\
&\,\le\, \Ltwo{x\big(T_{i}^{in}\big) - x^i\big(T_{i}^{in}\big)} +4\Rone \\
&\,=\, \Ltwo{x\big(T_{i}^{in}\big)- \ptraj\big(T_{i}^{in}\big) } +4\Rone .
\end{split}
\end{equation}

Combining (\ref{eq:dist increase throughout an in to out phase}) and (\ref{eq:dist increase throughout an out to in phase}) and a straightforward inductive argument, it follows that for any  $i\geq 0$, and for any $t\in[0,T_{i}^{out}]$ (where $T_{i}^{out}$ can be infinite), we have
\begin{equation} \label{eq:total dev at the end of round i}
\Ltwo{\ptraj(t) - x(t)} \le 5i\Rone.
\end{equation}


Let $\imax$ be  the maximum $i$ such that $T_i^{in}<\infty$, with the convention that $\imax=0$ if $T_1^{in}=\infty$. Equivalently, when $\imax$ is non-zero, it is the maximum $i$ such that the set on the right hand side of (\ref{eq:def Tin of pi}) is non-empty.  
We show that  $\imax$ is finite and in fact upper bounded by the number of critical points, $M$. In order to draw a contradiction, suppose that 
$\imax\geq M+1$. 
 In this case, there is a repeated critical point $p\in\CP$ that is effective in at least two phases; that is, there exist $i$, $j$, and $p$, such that 
$1\le i<j\le M+1$, $T^{in}_j<\infty$, and $p^i=p^j=p$.
Then,
\begin{equation}\label{eq:x , p less that 5m+2R1}
\begin{split}
\Ltwo{x\big(T_i^{in}\big) - p} &\,\le\, \Ltwo{x\big(T_i^{in}\big) - \ptraj\big(T_i^{in}\big) } + \Ltwo{\ptraj\big(T_i^{in}\big) - p^i} \\
&\,\le\, 5i\Rone + \Rone \\
&\,<\, 5(M+2)\Rone,
\end{split}
\end{equation}
where the second  inequality is due to (\ref{eq:total dev at the end of round i}) and the definition of $T_i^{in}$. 
The same bound also holds for $\Ltwo{x\big(T_j^{in}\big) - p}$. On the other hand, 
\begin{equation}\label{eq:3alpha bound}
\begin{split}
\Ltwo{x\big(T_i^{out}\big) - p}  &\,\ge\, \Ltwo{p^i - \ptraj\big(T_i^{out}\big)} - \Ltwo{x\big(T_i^{out}\big) - \ptraj\big(T_i^{out}\big)}\\
&\,\ge\,  \Rtwo- 5i\Rone\\
&\,\ge\,  \frac{\basinmin}{2}- 5(M+1)\Rone\\
&\,=\, 20(M+2)(\gamma+1)\saferad\ctresh -  5(M+1)\Rone\\
&\,\ge\, 20(M+2)(\gamma+1)\saferad\maxpert -  5(M+1)\Rone\\
&\,=\, 20(M+2)\Rone -  5(M+1)\Rone\\
&\,>\, 15(M+2)\Rone,
\end{split}
\end{equation}
where the second and third inequalities are due to  (\ref{eq:total dev at the end of round i}) and the definition of $\Rtwo$, respectively.
Let $\alpha= 15(M+2)\Rone$. 
Using the Definitions of $\ctresh$ and $\Rone$,  and the assumption $\theta<\theta^*$, we see that $\alpha<\basinmin$.
From Lemma \ref{lem-cp: basin cnc}, $\alpha$ is a basin radius for every  critical point. 
It follows from (\ref{eq:x , p less that 5m+2R1}), (\ref{eq:3alpha bound}), and again (\ref{eq:x , p less that 5m+2R1}) that
$\Ltwo{x\big(T_i^{in}\big) - p}<{\alpha}/{3}$,  $\Ltwo{x\big(T_i^{out}\big) - p}>\alpha$, and  $\Ltwo{x\big(T_j^{in}\big) - p}<{\alpha}/{3}$, respectively. Since $T_i^{in}\le T_i^{out} \le T_j^{in}$, this contradicts Lemma \ref{lem-cp: no revisit}. 
Therefore, the initial hypothesis $\imax\ge M+1$ cannot be true, and we conclude that $\imax\le M$. Hence, $T_{M+1}^{out}=\infty$. It follows from (\ref{eq:total dev at the end of round i}) and the definition of $\divconst$ in (\ref{eq:def divconst in the proof of cont th}) that for any $t\ge 0$,
\begin{equation}
\Ltwo{\ptraj(t) - x(t)} \,\le\, 5(M+1)\Rone \,\le\, \divconst \maxpert,
\end{equation}
which shows that the Theorem also holds for Case 2.

We now prove the last statement in Theorem \ref{th:main cont}, namely, that  the bound \eqref{eq:bounded pert cont} applies to the 
 system $\dot{x}\in F(x)+\lambda$,  with the same constant  $\divconst$.
Let $F'(\cdot) = F(\cdot)+\lambda$.
According to Lemma  \ref{lem-cp: F' vs F}, the systems $F$ and $F'$ have identical  effective regions and critical points. It follows that the constants $\basinmin$, $D^\CP$, $M$ (number of critical points), and $m$ (number of regions) 
are also the same. As a consequence,  the constants $\gamma$  and $\sigma$, defined in Lemmas \ref{lem:distance from cp} and \ref{lem:locally close}, respectively, are also identical for the two systems, 
and the same is true for the constant $\saferad = m2^{m+1}\sigma$, defined in Proposition \ref{prop:close to fluids at far}, 
the constant  $\ctresh$ defined in
(\ref{eq:ctresh}), and finally for the constant $\divconst$  in (\ref{eq:def divconst in the proof of cont th}). In other words,  the same constant $\divconst$ also works for the dynamical system $F'$. This completes the proof of Theorem \ref{th:main cont}.
\qed
\medskip



\medskip
\section{\bf Discussion} \label{sec:discussion}
In this section we review our main results and their implications, and also discuss the extent to which they can or cannot be generalized to broader classes of systems.

We have established a \emph{bounded input sensitivity} property of FPCS (finitely piecewise constant subgradient) systems, in a strong sense. In particular, we have shown that the increase in the distance between perturbed and unperturbed trajectories is upper-bounded 
by a constant multiple of the \emph{magnitude of the integral} 
 of the instantaneous perturbations; cf.~\eqref{eq:bounded pert cont}. 
As discussed in the introduction, this is much stronger than the elementary upper bounds which involve the integral of the magnitude of the instantaneous perturbations.
As an example, for the system illustrated in Fig. \ref{fig:application finite part is subdiff}, and with i.i.d.\  Bernoulli perturbations, the naive bound in \eqref{eq:dt2} grows at the rate of $t/2$, whereas the bound in \eqref{eq:bounded pert cont} only grows as $
(C/2)\sqrt{t}\log t$, for some $C<2000$, 
with high probability.
Thus, over short time scales, the naive bound is stronger, but in the regime of large $t$ (which is the relevant one for heavy-traffic asymptotic analysis) our bound is tighter. 
Furthermore, the best constant $C$ for that example is likely to be much smaller.\footnote{Using a variant of our proof, tailored to that example, it can be shown that $C$ can be set to 6.5.}
In any case, our work  carries out the important first step, that of showing that $C$ is in fact finite. We finally note that
 our definitions are broad enough to include as possible perturbations the sample paths of jump or Brownian motion processes.


\subsection{Implications}
FPCS  systems arise in many contexts. As discussed in Section~\ref{s:intro}, 
a prominent example is the celebrated Max-Weight policy for scheduling in queueing networks.  
Having made this connection, we can (cf.~\cite{AlTG18p2})
 apply a variant of our result to the Max-Weight policy, establish bounds on the distance between  the actual discrete-time stochastic system and its fluid approximation, and also obtain  state space collapse results 
 that are stronger
 than  available ones \cite{ShahW12,ShahTZ10}.

More broadly, flows or algorithms that evolve along the subgradient of a potential function are a fairly natural model, likely to arise in many other contexts. Recall also that, as mentioned ain Section \ref{s:intro}, the FPCS class has been shown \cite{AlTG18p4} to contain 
all non-expansive 
finite-partition hybrid systems that obey some minimal well-formedness and uniqueness properties.

\subsection{Generalizations}

Broad generalizations that assume only a subset of the properties of FPCS systems are not possible. In  
\cite{AlTG18p3}
we provide (counter)examples that show that a sensitivity bound of the form (6) does {\bf not} hold for various classes of systems. Our counterexamples include:

\begin{enumerate}
\item A non-expansive system; hence the non-expansiveness property is not sufficient by itself. 
\item A system that moves along the gradient of a twice continuously differentiable strictly convex function; hence the subgradient property is not sufficient by itself.
\item A system that moves along the subgradient of a piecewise linear convex function with infinitely many number of pieces; hence the finiteness of the number of pieces is essential. 
\end{enumerate}

Even though our main result cannot be extended by weakening its assumptions, it may still be possible to derive similar sensitivity bounds for other classes of systems. For example, 
\cite{AlTG18p3} provides necessary and sufficient conditions for linear systems
 $\dot x = Ax$, in terms of the spectrum of $A$. 
It will be interesting to explore whether there are some other natural classes of systems that do not have the non-expansiveness property but for which the conclusions in Theorem~\ref{th:main cont} are valid.

 \subsection{Some open problems}
 
 Besides attempts to obtain bounded sensitivity results for other types of systems, there are some interesting open problems for  FPCS systems specifically.
   
 \begin{enumerate}
 \item The bound in Theorem \ref{th:main cont} involves a constant $\divconst$
 which grows exponentially with the number of regions. It is not known whether this is unavoidable or whether a smaller (polynomial) constant is possible.
\item Theorem \ref{th:main cont} studies the distance between  
 a perturbed and an unperturbed trajectory, but this does not necessarily provide a strong bound on the distance between two perturbed trajectories.
Consider an FPCS system and two different perturbations $U_1(\cdot)$ and $U_2(\cdot)$ that are close at all times. We conjecture that in such a case the perturbed  trajectories are also close.
\end{enumerate}

\hide{
\subsection{Open Problems} \label{subsec:disc open}
Although non-expansiveness is  a necessary condition for the bound in (\ref{eq:diss pwc cont}), it is not necessary for the bound of type (\ref{eq:diss linear}), where the two trajectories have the same initial points. For instance, consider a two-dimensional linear dynamical system $\dot{x}=Ax$ with 
\begin{equation}
A=\left[  \begin{array}{lr}  -1/3 & 1\\ 0 & -1/3 \end{array}   \right].
\end{equation}
Since $A$ has two  negative real eigenvalues (equal to $-\frac13$), the system is SOF and Theorem \ref{th:linear} implies (\ref{eq:diss linear}). However, this system is not non-expansive, because if $x=0$ and $y=[1,1]^T$, then
\begin{equation}
(y-x)^T \big(Ay-Ax\big) \,=\, y^TAy \,=\,  \left[  \begin{array}{lr}   1 & 1 \end{array}   \right]  \left[  \begin{array}{lr}  -\frac13 & 1\\ 0 & -\frac13 \end{array}   \right] \left[  \begin{array}{c}   1 \\ 1 \end{array}   \right] \, =\, \frac13 \,>0.
\end{equation}
Therefore,  (\ref{eq:dir x (x-y) <0}) does not hold, and the system cannot be non-expansive.
 
In the following we discuss some open problems and future research directions.

\begin{openprob} [Deviation of Two Perturbed Trajectories in FPCS Systems]
In Theorem \ref{th:main cont}, we proved (\ref{eq:diss pwc cont}) for FPCS dynamical systems, which bounds the deviation of a perturbed trajectory from an unperturbed trajectory. The question is whether or not a bound of type (\ref{eq:diss linear}) holds for FPCS dynamical systems, bounding the deviation between two perturbed trajectories. 
As an application example, let $x(t)$ be a trajectory of a finite-partition control system with exogenous input $a(t)$, and let $\ptraj(t)$ be a perturbed trajectory with the same input $a(t)$ added by a small perturbation $\pert(t)$. Here, we are interested in a bound on $\Ltwo{\ptraj(t)-x(t)}$. Theorem \ref{th:main cont} provides a bound only for the case that $a(t)=\lambda$ is constant.
We conjecture that a bound of type  (\ref{eq:diss linear}) also holds for FPCS dynamical systems.
\end{openprob}

\begin{openprob} [Extension to More General Systems]
A major direction for future research is extending the input sensitivity bounds to more general classes of dynamical systems. Earlier discussions in this section, show that contractive property or subgradient fields are not sufficient conditions for (\ref{eq:diss pwc cont}) or (\ref{eq:diss linear}). In particular, we ask the question whether these bounds hold for  the class of piecewise linear dynamical systems (with all pieces being SOF), or the class of strongly contractive dynamical systems, i.e., a system where $\frac{d}{dt}\Ltwo{x(t)-y(t)}\le - \alpha \Ltwo{x(t)-y(t)}$, for all unperturbed trajectories $x(t)$ and $y(t)$, and for some uniform constant $\alpha$. 
\end{openprob} 
 
\begin{openprob}[Bounded Region]
In this paper, we dealt with dynamical systems defined over the entire $\R^n$. This is however not the case in many applications such as communication networks where the state space (i.e., queue lengths) is a subset of $\R^n$ with non-negative coordinates.  In this case, existence of boundary conditions further complicates the problem. Let $\Omega$ be the intersection of a finite number of half-spaces. For any point $x\in \Omega$ and any vector $v\in\R^n$, denote by $\pi^{\partial\Omega(x)}(v)$ the projection of $v$ on the conical hull of $\Omega-x$. Define the projected unperturbed and perturbed trajectories as:
\begin{equation}
\dot{x}(t) \in \pi^{\partial\Omega(x)}\big(F(x)\big),
\end{equation}
and \comment{see 'projected dynamical systems' in the literature}
\begin{equation}
\frac{d}{dt}\ptraj(t) \in \pi^{\partial\Omega(\ptraj)}\Big(F\big(\ptraj(t)\big) + \pert(t)\Big),
\end{equation}
respectively. We can ask the same question of  input sensitivity for these restricted dynamical systems, e.g., whether or not Theorem \ref{th:main cont}  holds for FPCS dynamical systems  restricted to $\Omega$. 
\end{openprob}

\begin{openprob} [Improved Bound]
Theorem \ref{th:main cont} establishes the existence of a constant $\divconst$ that satisfies (\ref{eq:diss pwc cont}). However, it does not provide a (polynomial) bound for $\divconst$. As such, although this theorem is a powerful tool for limit analyses, it fails to provide practical bounds for non-limit cases. 
 The question is whether there exist provable polynomial upper bounds for $\divconst$ or if one can prove exponential lower bounds.
\end{openprob}
}

\bibliography{schbib}
\ifIEEE
\bibliographystyle{ieeetr}
\else
\bibliographystyle{siamplain}
\fi

\ifOneCol
\else
\begin{IEEEbiography}[{\includegraphics[width=1in,height=1.25in,clip,keepaspectratio]{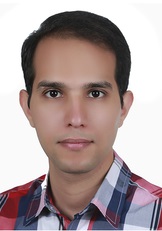}}]
{Arsalan Sharifnassab}
received a B.Sc. degree in Mathematics and also in Electrical Engineering, and an M.Sc. degree in Electrical Engineering, in 2011, 2011, and 2013 respectively, from the Sharif University of Technology, Iran, where he is currently a PhD student, again in EE.
He has been a visiting student in the Laboratory for Information and Decision Systems at MIT, over the academic year 2016-2017. His research interests include network scheduling and routing, optimization,  distributed algorithms, high-dimensional statistics, dynamic programming, and computational complexity.
\end{IEEEbiography}

\begin{IEEEbiography}[{\includegraphics[width=1in,height=1.25in,clip,keepaspectratio]{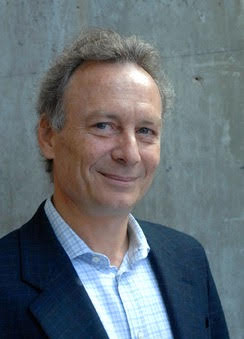}}]{John N. Tsitsiklis}
(S’80–M’81–SM’97–F’99) received the B.S. degree in mathematics and the B.S., M.S., and Ph.D. degrees in electrical engineering from the Massachusetts Institute of Technology (MIT), Cambridge, MA, USA, in 1980, 1980, 1981, and 1984, respectively. His research interests are in systems, optimization, communications, control, and operations research. He has coauthored four books and several journal papers in these areas.

He is currently a Clarence J. Lebel Professor with the Department of Electrical Engineering and Computer Science, MIT, where he serves as the director of the Laboratory for Information and Decision Systems. 
Among other distinctions, he is a member of the National Academy of Engineering, and a recipient of the ACM SIGMETRICS Achievement Award (2016), the IEEE Control Systems Award (2018), and the INFORMS von Neuman Theory Prize.
\end{IEEEbiography}

\begin{IEEEbiography}[{\includegraphics[width=1in,height=1.25in,clip,keepaspectratio]{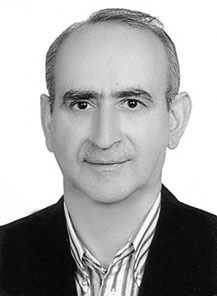}}]
{S. Jamaloddin Golestani} (M'85-SM'95-F'00) 
S. Jamaloddin Golestani (M’85-SM’95-F’00) was born in Iran in 1955. He completed his BSc in Electrical Engineering at the Sharif University of Technology in 1973, and his SM and PhD in Electrical Engineering and Computer Science at the Massachusetts Institute of Technology (MIT) in 1976 and 1979, respectively. 

Besides 21 years of academic work in Iran (1980 to 1988, and 2005 to present), he has spent 17 years (1988 to 2005) at Bell Communications Research, Morristown, NJ, and Bell Laboratories, Murray Hill, NJ. Since 2009, he has been with the Department of Electrical Engineering, Sharif University of Technology, Tehran, Iran. 

Dr. Golestani is the originator of the optimization approach to network flow control (1979) and the inventor of the self-clocked fair queueing (1994). He served on the editorial board of the IEEE/ACM Transactions on Networking from 1999 to 2001. He is the recipient of the IEEE Fellow award in 2000 with the citation ``for contributions to the theory of congestion control and provision of fairness and guaranteed services in packet networks''. His current research interests are network control and optimization, distributed computations, scheduling algorithms, and wireless networks. 
\end{IEEEbiography}
\fi


\ifOneCol
\newpage
\else
\medskip
\medskip
\medskip
\medskip
\medskip
\fi

\ifIEEE
\appendices
\section{\bf Proof of Lemma \ref{lem:cp}}\label{app:proof cp}
\else
\appendix
\section{\bf Proof of Lemma \ref{lem:cp}}\label{app:proof cp}
\fi
\begin{proof}[Proof of Lemma \ref{lem-cp: finite number}]
We fix some $\nu$ and  $p\in R_{\nu}=\{x \mid -\nu ^T x +b_{\nu} \geq -\mu ^T x+ b_{\mu},\ \forall\ \mu \}$. 
Suppose that $p$ is a critical point. 
By the definition of $\maxactset(p)$, we have $ -\nu ^T x+ b_{\nu}=-\mu ^T x +b_{\mu}$ for every $\mu\in \maxactset(p)$, i.e., these constraints are all active at $p$. 
Furthermore, by the definition of critical points, 
the vectors $\{\mu-\mu'\mid \mu,\mu'\in  \maxactset(p)\}$ span $\R^n$. It is not hard to see that this implies that the vectors
$\{\mu-\nu\mid \mu \in  \maxactset(p)\}$ also span $\R^n$, so that $n$ of them are linearly independent.
Using linear programming terminology, out of the constraints that define $R_{\nu}$, there are $n$ linearly independent active constraints at $p$, and 
$p$ is a ``basic feasible solution'' in $R_{\nu}$. This  is
equivalent to $p$ being an extreme point of $R_{\nu}$; cf.\ Theorem 2.3 in \cite{BertT97}.

For the converse implication, suppose that $p$ is an extreme point of $R_{\nu}$. Using again Theorem 2.3 in \cite{BertT97}, $n$ of the 
vectors $\mu-\nu$, associated with active constraints at $p$ (i.e., with $\mu\in 
\maxactset(p)$) are 
 linearly independent. It follows that the vectors $\mu-\mu'$, for $\mu,\mu'\in \maxactset(p)$ span $\R^n$, and (by the definition), $p$ is a critical point.
\end{proof}

\begin{proof}[Proof of Lemma \ref{lem-cp: basin straight movement}]
In order to draw a contradiction, consider a time $t>0$ where $z(t)$ is in the basin and $\dot{z}(t)\ne \dir$. It follows from Lemma \ref{lem:decreasing drift size} that $\Ltwo{\dot{z}(t)}<\Ltwo{\dir}$. Hence, $ \dir^T\dot{z}(t) \le \Ltwo{\dir} \cdot\Ltwo{\dot{z}(t)} < \Ltwo{\dir}^2$, which  contradicts  the definition of a basin.
\end{proof}

\medskip
\begin{proof}[Proof of Lemma \ref{lem-cp: basin infinite rad}]
We assume that the set $\CP$ of critical points is non-empty. 
We will first show that there exists a critical point $p$ such that $\Ltwo{\dir(p)}\le \Ltwo{\dir(x)}$, for all $x\in\R^n$. We will then show that $\R^n$ is a basin for 
this particular $p$.

Since there exists a critical point,  Part (a) implies that some $\region_{\nu}$ has an extreme point. 
Using linear programming theory (cf.~Theorem 2.6 in \cite{BertT97}) it follows that all of the non-empty regions $\region_{\mu}$ also have extreme points\footnote{This is because the regions are defined in terms of constraints $a^Tx\leq  b$ or $a^T x\geq b$, where each $a$ is of the form $a=\mu-\mu'$, for some $\mu,\mu'$; different regions correspond to different choices in the direction of the inequalities, but the vectors $a$ are the same or every region.}.

Consider some $x$ in some region $\region_{\nu}$. Let $x'$ be an extreme point of that region, chosen so that all constraints that were active at $x$ are also active at $x'$. (This can be done by moving inside $\region_{\nu}$ while respecting active constraints, until additional constraints are made active, exactly as in the proof of Theorem 2.6 in \cite{BertT97}.) The resulting extreme point $x'$ satisfies 
$\maxactset(x') \supseteq \maxactset(x)$. Since $F(x)$ is the convex hull of 
$\maxactset(x)$, it follows that $F(x') \supseteq F(x)$. From 
Lemma \ref{lem:min norm d+}, $\dir(x)$ is the minimum norm element of $F(x)$, which implies that $\Ltwo{\dir(x')}\le \Ltwo{\dir(x)}$. We conclude that when we minimize the function $\Ltwo{\dir(x)}$ over all $x\in\R^n$, it suffices to restrict to the (finite) set of extreme points of the different regions, or equivalently the set of critical points (cf. Part (a)). This concludes the proof  
that there exists a critical point, $p$, such that $\Ltwo{\dir(p)}\le \Ltwo{\dir(x)}$, 
for all $x\in\R^n$. Let $\dir^*=\dir(p)$.
 
We now proceed to show that $\R^n$ is a basin of $p$. Let $z(t)$ be an unperturbed trajectory with initial point $z(0)=p$. Similar to the proof of Part (b), if for some $t>0$, $\dot{z}(t) \ne \dir^*$, then it follows from Lemma \ref{lem:decreasing drift size} that $\Ltwo{\dir\big(z(t)\big)} = \Ltwo{\dot{z}(t)}<\Ltwo{\dot{z}(0)}=\Ltwo{\dir^*}$, which contradicts  the definition of $\dir^*$. Hence, for any $t\geq 0$, $z(t)=p+t\dir^*$,  
which implies that $-\dir^*\in  \partial \Phi \big(p+t\dir^*\big)$, where $\Phi$ is the convex function for which $F$ is the subdifferential.

For any $x\in\R^n$, 
let $\tilde{\Phi}(x) = \Phi(x)+t{\xi^*}^T(x-p)$, 
so that $\partial \tilde{\Phi}(x) = \partial \Phi(x) +\xi^*$.
Since $-\xi^*\in \partial {\Phi}(p+\dir^*)$, 
we have $0\in\partial \tilde{\Phi}\big(p+t{\xi^*}\big)$, which implies that $p+t{\xi^*}$ is a minimizer of $\tilde{\Phi}$. Consider an $x \in \R^n$ and a  $y \in F(x)$. Then, $y-\xi^*\in -\partial\tilde{\Phi}(x)$. It follows from the supporting hyperplane theorem that, for any $t\geq 0$, $(\xi^*-y)^T (p+t{\xi^*}-x) \le \tilde{\Phi}(p+t\xi^*) - \tilde{\Phi}(x)\le 0$, where the last inequality is because $p+t{\xi^*}$ is a minimizer of $\tilde{\Phi}$. 
Then, by letting $t$ go to infinity, we obtain
\begin{equation}
  \big(\dir^*-y \big)^T  {\dir^*}= \lim_{t\to\infty}  \frac{1}{t}\big(\dir^*-y \big)^T\big(p+t\dir^* -x\big) \le 0.
\end{equation}
Hence, $\Ltwo{{\dir^*}}^2\le {\dir^*}^T y$, 
which shows that $\R^n$ is a basin of $p$.
In the special case where $F$ is conic and has a critical point, then this is the only critical point and therefore has $R^n$ for a basin.
\end{proof}

\begin{proof}[Proof of Lemma \ref{lem-cp: basin cnc}]
Consider a critical point $p\in\CP$, and let
\begin{equation}\label{eq:def of phi tilde as the conic sys near cp}
\tilde{\Phi} (x) = \max_{\mu\in\maxactset(p)} \,\big\{-\mu^T\big(x-p\big)\big\},
\qquad \forall \ x\in\R^n.
\end{equation}
Hence, the dynamical system $\dot{x}\in \tilde{F}(x)\triangleq -\partial \tilde{\Phi}(x)$ is conic. 
Since the vectors $\big\{\mu-\mu' \mid \mu\in \maxactset(p)\big\}$ span $\R^n$, it follows that $p$ is also a critical point of the system $\dot{x}\in \tilde{F}(x)$. 
Lemma \ref{lem-cp: basin infinite rad} then implies that the entire set $\R^n$ is a basin for $p$, for the system $\dot{x}\in \tilde{F}(x)$. 

Let $\ball$ be  the ball of radius $\basinmin$ centred at $p$, where $\basinmin$ is the CNC. By the definition of the CNC, if $\ball\bigcap\region_\mu$ is non-empty for some $\mu\in\actionset$, then $p\in \region_\mu$. Hence, for any $x\in\ball$,  we must have  $\maxactset(x)\subseteq\maxactset(p)$. 
Therefore, for any $x\in\ball$,
\begin{equation}
\begin{split}
\Phi(x) &=  \max_{\mu\in\actionset} \big({-}\mu^Tx+b_\mu\big)\\
& =  \max_{\mu\in\maxactset(p)} \big({-}\mu^Tx+b_\mu\big)\\
& =  \max_{\mu\in\maxactset(p)} \Big({-}\mu^T\big(x-p\big) \,{-}\, \mu^Tp+b_\mu\Big)\\
&= \Phi(p)\,+\,\max_{\mu\in\maxactset(p)} \Big({-}\mu^T\big(x-p\big)\Big) \\
&= \Phi(p)\,+\, \tilde{\Phi} (x).
\end{split}
\end{equation}
where the second equality is because the set $\maxactset(x)$ of maximizers of $-\mu^Tx+b_\mu$  is a subset of $\maxactset(p)$.
Hence, for any $x\in\ball$, $F(x) = \tilde{F}(x)$. 
As a result,  for $x\in \ball$, $\dir(x)$ for the system $\dot{x}\in {F}(x)$ is equal to $\dir(x)$ for the system $\dot{x}\in \tilde{F}(x)$. 
Since $\R^n$ is a basin of $p$ for the system $\dot{x}\in \tilde{F}(x)$, it follows that for any $x\in \ball$ and any $y\in F(x) = \tilde{F}(x)$, we have $y^T\dir(p)\ge \Ltwo{\dir(p)}$. Hence, $\ball$ is a basin for the system $\dot{x}\in {F}(x)$.
\end{proof}

\begin{proof}[Proof of Lemma \ref{lem-cp: no revisit}]
The result will be derived by comparing the trajectory $x(t)$ of interest to another unperturbed trajectory, $z(t)$, initialized with $z(t_1)=p$. 
According to the non-expansive property of the dynamics, we have $\Ltwo{x(t)-z(t)}\le\Ltwo{x(t_1)-z(t_1)}\le {\basin}/{3}$, for every $t\ge t_1$. Hence,
\begin{equation*}
\Ltwo{z(t_2)-p} \,\ge\, \Ltwo{x(t_2)-p} - \Ltwo{x(t_2)-z(t_2)} \,>\, \basin - \frac{\basin}{3} \,=\, 2\basin/3.
\end{equation*}
In order to draw a contradiction, suppose that there is a time  $t_3>t_2$ such that $\Ltwo{x(t_3)-p}\le{\basin}/{3}$. In this case,
\begin{equation*} 
\Ltwo{z(t_3)-p} \,\le\, \Ltwo{z(t_3)-x(t_3)}  + \Ltwo{x(t_3)-p} \,\le\, \frac\basin3 + \frac{\basin}{3} \,=\, \frac23\basin.
\end{equation*}
Hence,  $z(t_3)$ is in the basin of $p$, which implies that $\dir(p)^T \dir\big(z(t_3)\big)\ge \Ltwo{\dir(p)}$ and
\begin{equation}\label{eq:speed of a return to basin is larger}
\Ltwo{\dir\big(z(t_3)\big)} \ge \Ltwo{\dir(p)}.
\end{equation}

The trajectory $z(t)$ starts inside the ${2\basin}/3$-neighbourhood of $p$ at time $t_1$, leaves this neighbourhood before time $t_2$, and returns back to it by time $t_3$. Since the ${2\basin}/3$-neighbourhood is convex, $z(t)$ must have changed its direction in the meanwhile, and there exists a time ${t'}\in(t_1,t_3)$ such that $\dot{z}({t'})\ne \dot{z}(t_1)=\dir(p)$.  Then, using Lemma \ref{lem:decreasing drift size},
\begin{equation} \label{eq:speed of a return to basin is smaller}
\Ltwo{\dir\big(z(t_3)\big)} \, \le \, \Ltwo{\dot{z}(t')} \, < \, \Ltwo{\dot{z}(t_1)} \, = \, \Ltwo{\dir(p)}.
\end{equation}
This contradicts (\ref{eq:speed of a return to basin is larger}) and concludes the proof.
\end{proof}

\begin{proof}[Proof of Lemma \ref{lem-cp: F' vs F}]   
For every drift $\mu\in \actionset$ of $F$, $\mu+\lambda$ is  a drift of $F'$. The associated effective region $\region'_{\mu+\lambda}$ of $F'$ is given by $\region'_{\mu+\lambda}=\big\{x\in\R^n\,\big| \, -(\mu+\lambda)^Tx + b_\mu\ge -(\nu+\lambda)^Tx + b_\nu,\, \forall \nu\in \actionset    \big\}=\big\{x\in\R^n\,\big| \, -\mu^Tx + b_\mu\ge -\nu^Tx + b_\nu,\, \forall \nu\in \actionset    \big\}=\region_\mu$.
Hence, the regions associated with $F$ and $F'$ are the same. Consider a point $p\in\R^n $ and let $\maxactset'(p)=\maxactset(p)+\lambda$ be the set of active drifts of $p$ in system $F'$. The affine span of $\maxactset'(p)$ is  $\R^n$ if and only if the affine span of $\maxactset(p)$ is $\R^n$. Hence, $p$ is a critical point for the system $F'$ if and only if it is a critical point for the system $F$. Finally, by the definition of the CNC, since $F$ and $F'$ have the same set of regions and the same set of critical points, they also have the same CNC.
\end{proof}


{

\newpage
\setcounter{page}{1}

\begin{center}
\Large \bf Supplementary Material
\end{center}

\ifIEEE
\section{\bf Proof of Lemma \ref{lem:distance from cp}} \label{app:proof gamma}
\else
\section{\bf Proof of Lemma \ref{lem:distance from cp}} \label{app:proof gamma}
\fi
We provide here the proof of Lemma \ref{lem:distance from cp}. We will make use of an auxiliary result,  proved in \cite{MannT05}, which states that if a point is close to each of several half-spaces, then that point is also close to the intersection of those half-spaces.
\begin{lemma}[\cite{MannT05}, Lemma 5.1]\label{lem:known lemma! distance from intersection of half-spaces}
Given a finite collection  of half-spaces 
$W_i\subset \R^n$, with non-empty intersection, there exists a finite constant $c>0$ such that 
\begin{equation}
d\Big(x\,,\,\bigcap_i W_i\Big) \leq c \cdot	 \max_i d\big(x\,,\,W_i\big),
\qquad \forall\ x\in \R^n.
\end{equation}
\end{lemma}

\begin{proof}[Proof of Lemma \ref{lem:distance from cp}]

For any $x\in \R^n$, let $r(x)=\sup\big\{r:\,\ldset_r\textrm{ is low-dimensional}\big\}$. By definition, if $x$ is not a critical point, then $r(x)>0$, and if $r\ge r(x)$, then $\ldset_r$ is not low-dimensional. We will show that
\begin{equation} \label{eq:gamma bar is positive}
\bar{\gamma}\triangleq \inf_{x\not\in\CP} \, \frac{r(x)}{d\big(x,\CP\big)}\,> \, 0.
\end{equation}
In order to draw a contradiction, suppose that there exists a sequence of points $y_k\in\R^n\backslash \CP$  such that 
\begin{equation}\label{eq:absurd hypothesis on yk go to 0}
\frac{r(y_k)}{d\big(y_k,\CP\big)}\xrightarrow[\,k\to\infty\,]{} 0.
\end{equation}
Since $\ldset_{r(y_k)}(y_k)$ is not low-dimensional, there exist $n+1$ drifts $\mu_1,\ldots,\mu_{n+1}\in \ldset_{r(y_k)}(y_k)$ such that 
\begin{equation}\label{eq:span of r neigh regions is n assumption}
 \mathrm{span }\big\{ \mu_i-\mu_j \, \big| \, i,j\le n+1 \big\} =\R^n.
\end{equation}
Because the set $\actionset$ of all drifts is finite, there exists an infinite subsequence $\big\{x_k\big\}$ of $\big\{y_k\big\}$ for which (\ref{eq:span of r neigh regions is n assumption}) holds for the same set of drifts. We fix this set of drifts $\big\{\mu_i\big\}_{i=1}^{n+1}$. Then, for any $k$, $\big\{\mu_i\big\}_{i=1}^{n+1}\subseteq \ldset_{r(x_k)}(x_k)$. It follows from the definition of $r(x)$ that for any $k$,
\begin{equation}\label{eq:rxk is the max of dist from Rj}
r(x_k)=\max_{i\le n+1 } d\big(x_k\,,\, \region_i\big),
\end{equation}
where $\region_i=\region_{\mu_i}$ is the effective region of $\mu_i$.
We define $n(n+1)$ half-spaces $W_{i,j}$ as follows. For any $i,j \le n+1$ with $i\ne j$, let
\begin{equation}
W_{i,j}\triangleq \left\{ x\in\R^n\,\,\Big|\,\, -(\mu_i-\mu_j)^Tx\,+\,b_i-b_j\,\ge\,0     \right\},
\end{equation}
where  $b_i$ is a shorthand for $b_{\mu_i}$.
Then, for any $i\le  n+1$,
\begin{equation}
\region_i \subseteq \bigcap_{j\ne i} W_{i,j}.
\end{equation}
Hence, for any $i\le n+1$ and any $x\in\R^n$, $d\big(x,\region_i\big)\ge \max_{j\le n+1} d\big(x,W_{i,j}\big)$. Then, it follows from (\ref{eq:rxk is the max of dist from Rj}) that for any $k\ge 1$,
\begin{equation}\label{eq:r(x) le max dist from wijs}
r(x_k)\,=\,\max_{i\le n+1 } d\big(x_k\,,\, \region_i\big)\, \ge \, \max_{\substack{i,j\le n+1\\ i\ne j}} d\big(x,W_{i,j}\big)
\end{equation}

It follows from (\ref{eq:span of r neigh regions is n assumption}) that the following system of $n$ linear equations is non-degenerate: 
\begin{equation}
-(\mu_i-\mu_{n+1})^Tx\,+\,b_i-b_{n+1}\,=\,0, \quad i=1,\ldots,n.
\end{equation}
Hence, it has a unique solution, which we denote by $p$. 
Note that $W_{i,j}$ and $W_{j,i}$ are different, and their intersection is $\big\{  x\,\,\big|\,\, -(\mu_i-\mu_j)^Tx\,+\,b_i-b_j\,=\,0\big\}$. Therefore,
\begin{equation}\label{eq:p is the intersec of wijs}
\{p\}=\bigcap_{\substack{i,j\le n+1\\ i\ne j}} W_{i,j}.
\end{equation}
It follows from Lemma \ref{lem:known lemma! distance from intersection of half-spaces}, with $\delta=1/c$, that there exists a constant $\delta>0$ such that for any $x\in\R^n$,
\begin{equation} \label{eq:max dist from wij is less than c times dist of xk and p}
\max_{\substack{i,j\le n+1\\ i\ne j}} d\big(x,W_{i,j}\big) \,\ge\,  \delta\, d\Big(x,\bigcap_{\substack{i,j\le n+1\\ i\ne j}} W_{i,j}\Big) \,=\, \delta\, d\big(x,p\big).
\end{equation}
Combining  (\ref{eq:r(x) le max dist from wijs})  and (\ref{eq:max dist from wij is less than c times dist of xk and p}), we have for any $k$,
\begin{equation}\label{eq:rxk less than c dist of xk and p}
r(x_k) \,\ge\, \max_{\substack{i,j\le n+1\\ i\ne j}} d\big(x_k,W_{i,j}\big) \,\ge\,  
\delta\, d\big(x_k,p\big).
\end{equation}

Back to the hypothesis (\ref{eq:absurd hypothesis on yk go to 0}), there are two possible cases: (a) $\big\{x_k\big\}$ has a subsequence $\big\{z_k\big\}$ with $d\big(z_k,\CP\big)\to \infty$, or (b) $x_k$ has a subsequence $z_k$ with $r(z_k)\to 0$.

In the first case, where $d\big(z_k,\CP\big)\to \infty$, it follows from (\ref{eq:rxk less than c dist of xk and p}) that
\begin{equation}
\begin{split}
\lim_{k\to\infty} \frac{r(z_k)}{d\big( z_k,\CP \big)}   &\ge \delta \lim_{k\to\infty} \frac{d\big( z_k,p \big)}{d\big( z_k,\CP \big)}\\
&\ge \delta  \lim_{k\to\infty} \frac{d\big( z_k,\CP \big) - d\big( p,\CP \big)}{d\big( z_k,\CP \big)}\\
&=\delta\,>\,0,
\end{split}
\end{equation}
which contradicts (\ref{eq:absurd hypothesis on yk go to 0}).

In the second case, where $r(z_k)\to 0$,  it follows from (\ref{eq:rxk less than c dist of xk and p}) and (\ref{eq:rxk is the max of dist from Rj}) that for any $i\le n+1$,
\begin{equation}
d\big(p,\region_i \big) \,\le\, d\big(p, z_k\big) + d\big(z_k,\region_i \big) \, \le\, \frac{r(z_k)}{\delta} + r(z_k)\, \xrightarrow[\,k\to\infty\,]{} 0.
\end{equation}
Then, since each $\region_i$ is a closed set, we must have $p\in\region_i$. Hence, $p\in\bigcap_{i\le n+1}\region_i$ which together with (\ref{eq:span of r neigh regions is n assumption}) implies that $p$ is a critical point. Using this fact, and then (\ref{eq:rxk less than c dist of xk and p}), we obtain
\begin{equation}
\lim_{k\to\infty} \frac{r(z_k)}{d\big( z_k,\CP \big)} \ge \lim_{k\to\infty} \frac{r(z_k)}{d\big( z_k,p \big)} \ge \delta > 0,
\end{equation}
which again contradicts (\ref{eq:absurd hypothesis on yk go to 0}). Hence,  (\ref{eq:absurd hypothesis on yk go to 0}) is contradicted in both cases, and (\ref{eq:gamma bar is positive}) follows.

Let $\gamma=\max\{1,\,{1}/{\bar{\gamma}}\}$. 
It follows from (\ref{eq:gamma bar is positive}) that $\gamma r(x)\ge d\big(x,\CP\big)$, for all $x\not\in\CP$. Hence, if $\gamma r < d\big(x,\CP\big)$, then $r<r(x)$, and by the definition of $r(x)$, $\ldset_r(x)$ is low-dimensional.
\end{proof}

}

\end{document}